\documentclass[reqno]{amsart}
\makeatletter
\def\l@subsection{\@tocline{2}{0pt}{2.5pc}{5pc}{}}
\def\l@subsubsection{\@tocline{2}{0pt}{5pc}{7.5pc}{}}

\usepackage{mathrsfs}
\usepackage{amsbsy}
\usepackage{frcursive}
\usepackage{yfonts}
\usepackage{graphicx}
\usepackage{empheq}
\usepackage{mathtools}
\usepackage{tikz}
\usepackage{amsbsy}
\usepackage{avant}
\usepackage{amsmath}
\usepackage{amsthm}
\usepackage{amssymb}
\usepackage{mathrsfs}
\usepackage{amsfonts}
\usepackage{newlfont}
\usepackage[sans]{dsfont}
\usepackage{dcolumn}
\usepackage{bm}
\usepackage{bbm}
\usepackage{stmaryrd}
\usepackage{bbm}
\usepackage{relsize}
\usepackage{euscript}
\usepackage{eufrak}
\usepackage{enumerate}
\usepackage[margin=1.0in]{geometry}
\numberwithin{equation}{section}
\newtheorem{thm}{Theorem}[section]

\newtheorem{cor}[thm]{Corollary}
\newtheorem{lem}[thm]{Lemma}
\newtheorem{prop}[thm]{Proposition}
\newtheorem{defn}[thm]{Definition}
\newtheorem{rem}[thm]{Remark}

\begin{document}
\allowdisplaybreaks{
\title[]{SECOND AND THIRD-ORDER STRUCTURE FUNCTIONS OF AN 'ENGINEERED' RANDOM FIELD AND EMERGENCE OF THE KOLMOGOROV 4/5 AND 2/3-SCALING LAWS OF TURBULENCE}
\author{Steven D Miller}\email{stevendm@ed-alumni.net}
\address{Rytalix Capital, Scotland.}
\maketitle
\begin{abstract}
The 4/5 and 2/3 laws of turbulence can emerge from a theory of 'engineered' random vector fields $\mathcal{X}_{i}(x,t) =X_{i}(x,t)+\tfrac{\theta}{\sqrt{d(d+2)}} X_{i}(x,t)\psi(x)$ existing within $\mathbf{D}\subset\mathbf{R}^{d}$. Here, $X_{i}(x,t)$ is a smooth deterministic vector field obeying a nonlinear PDE for all $(x,t)\in\mathbf{D}\times\mathbf{R}^{+}$, and $\theta$ is a small parameter. The field $\psi(x)$ is a regulated and differentiable Gaussian random field with expectation $\mathbb{E}[\psi(x)]=0$, but having an antisymmetric covariance kernel $\mathscr{K}(x,y)=\mathbb{E}[\psi(x)\psi(y)]=f(x,y)K(\|x-y\|;\lambda)$ with $f(x,y)=-f(y,x)=1,f(x,x)=f(y,y)=0$ and with $K(\|x-y\|;\lambda)$ a standard stationary symmetric kernel. For $0\le\ell\le \lambda<L$ with $X_{i}(x,t)=X_{i}=(0,0,X)$ and $\theta=1$ then for $d=3$, the third-order structure function is
\begin{align}
S_{3}[\ell]=\mathbb{E}\left[|\mathcal{X}_{i}(x+\ell,t)-\mathcal{X}(x,t)|^{3}\right]=-\frac{4}{5}\|X_{i}\|^{3}=-\frac{4}{5}X^{3}\nonumber
\end{align}
and $S_{2}[\ell]=CX^{2}$. The classical 4/5 and 2/3-scaling laws then emerge if one identifies the random field $\mathcal{X}_{i}(x,t)$ with a turbulent fluid flow $\mathcal{U}_{i}(x,t)$ or velocity, with mean flow $\mathbb{E}[\mathcal{U}_{i}(x,t)]=U_{i}(x,t)=U_{i}$ being a trivial solution of Burger's equation. Assuming constant dissipation rate $\epsilon$, small constant viscosity $\nu$, corresponding to high Reynolds number, and the standard energy balance law, then for a range $\eta\le\ell\ll \lambda<L$
\begin{align}
S_{3}[\ell]=\mathbb{E}\left[|\mathcal{U}_{i}(x+\ell,t)-\mathcal{U}(x,t)|^{3}\right]=-\frac{4}{5}\epsilon\ell\nonumber
\end{align}
where $\eta=(\nu^{3/4}\epsilon)^{-1/4}$. For the second-order structure function, the 2/3-law emerges as $S_{2}[\ell]=C\epsilon^{2/3}\ell^{2/3}$.
\end{abstract}
\tableofcontents
\raggedbottom
\maketitle
\section{{INTRODUCTION:~THE 4/5 AND 2/3 SCALING LAWS OF TURBULENCE}}
In fluid mechanics, the 4/5 and 2/3-laws and the law of finite energy dissipation are very important and well-established foundational results of modern turbulence theory, and there is by now a very considerable volume of literature devoted to them in both theory and experiment, and to turbulence in general \textbf{[1-57]}, and references therein]. However, there is no mathematically rigorous and fully deductive theory or derivation which begins with the Navier-Stokes equations and derives these laws exactly from first principles. In this note, it is shown how the structure or exact form of both laws can emerge naturally from a theory of specifically 'engineered' forms of random vector fields existing within a d-dimenional domain. The classical 4/5 and 2/3 scaling laws then follow naturally if one identifies these random vector fields with a turbulent fluid flow or velocity, assuming the fluid has a constant dissipation rate $\epsilon$ and a constant (very small) viscosity $\nu$ and obeys the standard energy balance law.

Suppose an incompressible fluid with very small viscosity $\gtrapprox 0$ and velocity $U(x,t)$, evolving by the Burgers or Euler equations from some initial data and with suitable boundary conditions, flows within $\mathbf{D}\subset\mathbf{R}^{d}$. If the fluid velocity is dissociated into 'mean' and 'fluctuating' contributions then $U(x,t)=\overline{U(x,t)}+\widetilde{U(x,t)}$ with $\langle U(x,t)\rangle=\overline{U(x,t)}$ with $\langle\widetilde{U(x,t)}\rangle=0$ and $\langle U(x,t)U(y,t)\rangle\ne 0$. Here $\langle\bullet\rangle$ is a suitable average such as a time or ensemble average. Then at very high, but not infinite, Reynolds number $\mathscr{R}\gg 0 $ all of the small-scale statistical properties are assumed to be uniquely and universally determined by the length scale ${\ell}$, the mean dissipation rate (per unit mass). Despite its conjectural status from the perspective of mathematical rigour, with some heuristic assumptions on statistical properties (homogeneity, isotropy), Kolmogorov \textbf{[1-5]}
made these key prediction about the structure of turbulent velocity fields for incompressible viscous fluids at high Reynolds number, namely that for $d=3$, the following 4/5 and 2/3-scaling law holds over an inertial range
\begin{align}
&{S}_{3}[\ell]=\big\langle\big|{U({x}+{\ell},t)}-{U({x},t)}\big|^{3}\big\rangle=-\frac{4}{5}\mathlarger{\epsilon}\ell\\&
{S}_{2}[\ell]=\big\langle\big|{U({x}+{\ell},t)}-{U({x},t)}\big|^{2}\big\rangle
={C}_{2}\mathlarger{\epsilon}^{2/3}{\ell}^{2/3}
\end{align}
where ${C}_{2}$ is some constant. These should hold in the limit of large Reynolds number and small scales $\ell\sim 0$. In particular, the 4/5-law is an exact result. In Fourier space, the 2/3-law becomes the $5/3-law$ for the energy spectrum. Here, ${\ell}$ is within the so-called
\textit{inertial range} of length scales
\begin{align}
\eta \le {\ell} \le {L}
\end{align}
The length $\eta=(\nu^{3/4}\epsilon)^{-1/4}$ known as the Kolmogorov scale, represents a small scale dissipative cutoff or the size of the smallest eddies, and the integral scale L represents the size of the largest eddy in the flow, which cannot exceed the dimensions of the domain. At this scale, viscosity dominates and the kinetic energy is dissipated into heat. There is a 'cascade' process whereby energy is transferred from the largest scales/eddies to the Kolmogorov scale. The objects ${S}_{p}[\ell]$ are the pth-order (longitudinal) structure functions and have the generic form
\begin{align}
& S_{p}[\ell]=\big\langle\big|{U(x+{\ell},t)}-{U({x},t)}\big|^{p}\big\rangle
={C}_{p}\big({\epsilon}{\ell}~\big)^{\zeta_{p}}
\end{align}
Although highly cited, the work still remains mathematically incomplete and essentially heuristic. But the central results have remained robust over the decades and are at least acceptably correct within the confines of the underlying assumptions, and supported by strong experimental evidence, at least under specific conditions. But it is also well known that real fluids in general do not conform exactly to the Kolmogorov predictions. Intermittency, non-uniformity of the velocity’s 'roughness' and energy dissipation rate, result in deviations of the scaling exponents from a purely linear behavior in \textbf{[6-10]}. Experiments do however indicate that for p near three, the formula approximately holds with $\zeta_{2}=\tfrac{2}{3}+[0.03, 0.06]$ and $\zeta_{3}\sim 1$. For example, in the flow past a sphere a value of $\zeta_{2}\sim 0.701$ is reported in $\textbf{[11]}$ and $\zeta_{2}\sim 0.71$ in \textbf{[12]}. Some recent high-resolution numerical simulations report $\zeta_{2}\sim 0.725$. Although there are slight variations, these results all conform to $\zeta_{2}\gtrapprox\tfrac{2}{3}$ and $\zeta_{3}\lessapprox 1$. Kolmogorov also improved the 2/3-law in 1962 to account for intermittency \textbf{[3]}.

It is clear that fluid mechanics continues to be very challenging, both physically and computationally, and from the perspective of mathematical rigour. Many of the issues discussed by Von Neumann in his well-known review paper remain relevant \textbf{[26]}. There remains opportunity (and an ongoing need) to apply new and established mathematical tools and methods to the problem of developed turbulence: these should include stochastic PDE, stochastic and statistical/random geometry and random fields. A central issue within fully developed turbulence is how to define and calculate Reynolds stresses, structure functions and velocity correlations. Established methods are mostly heuristic and it is very difficult to rigorously define or mathematically formalise the required spatial, temporal or ensemble averages $\big\langle\bullet\big\rangle$ in a useful manner. Rigorously defining statistical averages in conventional statistical hydrodynamics is fraught with technical difficulties and limitations, as well as having a limited scope of physical applicability. But a key insight of Kolmogorov's work is that turbulent flows seem to be essentially random fields.

In this paper, we consider a mathematical construction of spatio-temporal random vectors fields within a closed Euclidean domain $\mathbf{D}\subset\mathbf{R}^{d}$. The classical 4/5 and 2/3-scaling laws then emerge for $d=3$ if one identifies the random field with an 'engineered' turbulent fluid flow or velocity, with mean flow $U_{i}(x,t)=U_{i}$ being a trivial solution of Burger's equation. And assuming constant dissipation rate $\epsilon$, small constant viscosity $\nu\sim 0$, corresponding to high Reynolds number, and the standard energy balance law. These scaling laws hold for a range $\eta\le\ell\ll \lambda<L$, where $\lambda$ is a correlation length.
\section{{RANDOM SCALAR FIELDS AND 'ENGINEERED' RANDOM VECTOR FIELDS IN A DOMAIN $\mathbf{D}\subset\mathbf{R}^{d}$}}
In this section, the 3rd order and 2nd-order structure functions of a random vector field are calculated, without any reference to fluid mechanics. Later, it will be assumed that the noise or random fluctuation in fully developed turbulence is a generic noise determined by general theorems in probability theory, stochastic analysis, and random fields or functions. Classical random fields or functions correspond naturally to structures, and properties of systems, that are varying randomly in time and/or space. They have found many useful applications in mathematics and applied science: in the statistical theory or turbulence, in geoscience, machine learning and data science, medical science, engineering, imaging, computer graphics, statistical mechanics and statistics, biology and cosmology $\mathbf{[58-84]}$. Gaussian random fields (GRFs) are of special significance as they are more mathematically tractable and can occur spontaneously in systems with a larger number of degrees of freedom via the central limit theorem. A GRF is defined with respect to a probability space/triplet as follows:
\begin{defn}(\textbf{Formal definition of Gaussian random fields})\newline
Let $({\Omega},\mathscr{F},{\mathbb{P}})$ be a probability space. Then: $\mathbb{P}$ is a function such that $\mathbb{P}:\mathscr{F}\rightarrow [0,1]$, so that for all $\mathcal{B}\in\mathscr{F}$, there is an associated probability $\mathbb{P}(\mathcal{B})$. The measure is a probability measure when $\mathbb{P}(\Omega)=1$. Let $x_{i}\subset\mathbf{D}\subset{\mathbf{R}}^{n}$ be Euclidean coordinates and let $(\Omega,\mathscr{F},\mathbf{P})$ be a probability space. Let $\mathscr{F}(x;\omega)$ be a random scalar function that depends on the coordinates $x\subset{\mathbf{D}}\subset{\mathbf{R}}^{n}$ and also $\omega\in {\Omega}$. Given any pair $(x,\omega)$ there ${\exists}$ map $\mathfrak{M}:{\mathbf{R}}^{n}\times\Omega\rightarrow{\mathbf{R}}$ such that $\mathfrak{M}:(\omega,x)\longrightarrow\psi(x;\omega)$, so that $\psi(x;\omega)$ is a \textbf{random variable or field} on $\mathbf{D}\subset\mathbf{R}^{n}$ with respect to the probability space $(\bm{\Omega},\mathscr{F},\mathbf{P})$. A random field is then essentially a family of random variables $\lbrace \psi(x;\omega)\rbrace$ defined with respect to the space $(\Omega,\mathscr{F},\mathbf{P})$ and ${\mathbf{R}}^{n}$. The fields can also include a time variable $t\in{\mathbf{R}}^{+}$ so that given any triplet $(x,t,\omega)$ there is a mapping $\mathfrak{M}:{\mathbf{R}}\times{\mathbf{R}}^{n}\times\bm{\Omega}\rightarrow {\mathbf{R}}$ such that $\mathfrak{M}:(x,t,\omega)\hookrightarrow \psi({x},t;\omega)$ is a \textbf{spatio-temporal random field}. Normally, the field will be expressed in the form $\psi(x,t)$ or $\psi({x})$ with $\omega$ dropped. From here, only spatial fields $\psi(x)$ will be considered. The random field $\psi(x)$ will have the following bounds and continuity properties [REFs] $\mathbb{P}[\sup_{x\in\mathbf{D}}|\psi({x})|~~<~~\infty]~=+1$ and
$\mathbb{P}[\lim_{{x}\rightarrow {y}}\big|\psi({x})-{\psi}(x)\big|=0,~\forall~({x},{y})\in\mathbf{D}]=1 $.
\end{defn}
\begin{lem}
The random field is at the least, mean-square differentiable in that \textbf{[56], [62], [64]}
\begin{align}
\nabla_{j}\psi(x)=\frac{\partial}{\partial x_{j}}\Xi(x)= \lim_{{\ell}\rightarrow 0} \big\lbrace{\psi(x+|{\ell}|\mathlarger{\bm{\mathsf{e}}}_{j})-\psi(x)}\big\rbrace{|{\ell}|^{-1}}
\end{align}
where $\mathlarger{\bm{\mathsf{e}}}_{j}$ is a unit vector in the $j^{th}$ direction. For a Gaussian field, sufficient conditions for differentiability can be given in terms
of the covariance or correlation function, which must be regulated at ${x}={y}$ The derivatives of the field $\nabla_{i}\psi,\nabla_{i}\nabla_{j}\psi({x})$ exist at least up to 2nd order and do line, surface and volume integrals $\mathlarger{\int}_{\bm{\Omega}}\psi(x,t)d\mu(x)$. The derivatives or integrals of a random field are also a random field.
\end{lem}
\begin{defn}
The stochastic expectation ${\mathbb{E}}\langle\bullet\rangle $) and binary correlation with respect to the space $(\Omega,{\mathscr{F}},{{\mathbb{P}}})$ is defined as follows, with $(\omega,\vartheta)\in{\Omega}$
\begin{align}
{\mathbb{E}}[\bullet]=\mathlarger{\int}_{\omega}\bullet~d\bm{\mathbb{P}}[\omega],~~~~
{\mathbb{E}}[\bullet{\times}\bullet]=\mathlarger{\int}\!\!\!\!\mathlarger{\int}_{\Omega}
\bullet{\times}\bullet~d\bm{\mathbf{P}}[\omega]d\bm{\mathbb{P}}[\vartheta]
\end{align}
For Gaussian random fields $\bullet=\psi(x,t)$ only the binary correlation or covariance is required so that
\begin{align}
&{\mathbb{E}}[\psi(x)]=\mathlarger{\int}_{\omega}\psi(x;\omega)~d{\mathbb{P}}[\omega]=0\nonumber\\&
{\mathbb{E}}[\psi(x)\psi(y)]=
\mathlarger{\int}\!\!\!\!\mathlarger{\int}_{\Omega}\psi(x;\omega)
\psi(y;\zeta)~d{\mathbb{P}}[\omega]d{\mathbb{P}}[\zeta]=K(\|x-y\|;\lambda)\bm{\varphi}(t,s)
\end{align}
and regulated at ${x}={y}$ for all $({x},{y})\in\mathbf{D}$ and $t\in[0,\infty)$ if $\mathbb{E}[\psi(x)\psi(x)]<\beta<\infty$.
\end{defn}
\begin{defn}
Two random fields $\psi({x}),\psi({y}))$ defined for any $({x},{y})\in\mathbf{D}$ are correlated or uncorrelated if
\begin{empheq}[right=\empheqrbrace]{align}
&{\mathbb{E}}[\psi(x)\psi(y)]\ne 0 \nonumber\\&
{\mathbb{E}}[\psi(x)\psi(y)]= 0
\end{empheq}
\end{defn}
\begin{defn}
The covariance function of a zero-centred Gaussian random field is
\begin{align}
cov\left(\psi(x),\psi(y)\right)=\bm{\mathbb{E}}[\psi({x})\psi({y})]
+\bm{\mathbb{E}}[\psi({x})]\bm{\mathbb{E}}[\psi(y)]=\bm{\mathbb{E}}[\psi({x})\psi({y})]=K(\|x-y\|;\lambda)
\end{align}
so that the binary correlation and the covariance are equivalent. Here $\lambda$ is the correlation length. The GRF is isotropic if $K(\|x-y\|;\lambda)=K(\|y-x\|;\lambda)$ depends only on the separation $\|x-y\|$ and is stationary if $K(\|(x+\delta x)-(y+\delta y)\|;\lambda)=K(\|x-y\|;\lambda)$. Hence, the 2-point function $K(\|x-y\|;\lambda)$ is translationally and rotationally invariant in $\mathbf{R}^{d}$ for all $\delta x>0$ and $\delta y>0$.
\end{defn}
Typical covariances for Gaussian random fields $\psi(x)$ are the rational quadratic form
\begin{align}
{\mathbb{E}}[\psi(x)\psi(y)]=K(\|x-y\|;\lambda)=\beta\left(1+\frac{\|x-y\|^{2}}{2\alpha\lambda^{2}}\right)^{-\alpha}
\end{align}
where $\lambda$ is the correlation length and $\alpha$ is the 'scale-mixing' parameter. Another commonly used covariance kernel is the Gaussian
\begin{align}
{\mathbb{E}}[\psi({x})]\psi(y)]=K(\|x-y\|;\lambda)=\beta\exp\left(-\frac{|x-y|^{2}}{\lambda^{2}}\right)
\end{align}
In both cases, in the limit that $\lambda\rightarrow 0$, the noise reduces to a white-in-space noise which is delta correlated
\begin{align}
{\mathbb{E}}[\psi(x)\psi(y)]\longrightarrow{\mathbb{E}}[\mathcal{W}(t)\mathcal{W}(s)]=-\delta^{3}(x-y)
\end{align}
and the 2nd-order moment blows up in that ${\mathbb{E}}[|\psi({x})|^{2}]=\infty$. The random field or noise $\psi(x)$ is differentiable because the field is regulated in that
\begin{align}
{\mathbb{E}}[\psi({x})\psi(x)]=\mathlarger{\phi}(x,x;\lambda)=\mathlarger{\beta}<\infty
\end{align}
Gaussian random fields also have a Fourier representation in ${k}$-space.
\begin{defn}
If $\mathfrak{F}:\mathbf{R}\rightarrow\mathbf{K}$ is a Fourier transform then a generic random
Gaussian scalar field $\psi(x)$ is said to be \textbf{harmonisable} if
\begin{align}
{\psi}(x)={\int}_{\mathbb{R^{3}}}\exp(i{k}_{i}{x}^{i})\psi(k)d^{3}k
\end{align}
Let $\psi(x)$ be an arbitrary harmonisable Gaussian random scalar field existing for all $x\in{\mathbf{R}}^{3}$. Given the basic Fourier representation of the binary correlation
\begin{align}
{\mathbb{E}}[\psi(x)\psi(y)]=\mathlarger{\int}_{\mathbb{K}^{3}}d^{3}k{\Phi}(k)\exp(i k_{i}(x-y)^{i})
\end{align}
where ${\Phi}(k)$ is a spectral function, then for ${x}={y}$ such that $\mathbb{E}\left[{\psi}({x})~{\psi}({x})\right]={\int}_{\mathbb{K}^{3}}{\Phi}(k)d^{3}k $.
\end{defn}
For ${\Phi}(k)=1$ for example, one recovers an unregulated white noise with
\begin{align}
{\mathbb{E}}[\psi(x)\psi(y)]={\int}_{\mathbb{K}^{3}}d^{3}k
\exp(i k_{i}(x-y)^{i})={\mathbb{E}}[\mathcal{W}(x)~\mathcal{W}(y)]
={C}\delta^{3}(x-y)
\end{align}
and for $\Phi(k)=\tfrac{\beta}{k^{2}}\exp\left(-\tfrac{1}{4}\lambda^{2}k^{2}\right)$, one recovers the kernel (2.7).

It is possible to construct or 'engineer' new kernels from the standard kernels. (This kernel will be applied later.)
\begin{prop}
Let $\psi:\mathbf{D}\times\Omega\rightarrow\mathbf{R}^{+}$ be a GRF existing for all $x\in\mathbf{D}$.
Let $f:\mathbf{D}\times\mathbf{D}\rightarrow \lbrace -1,0,1\rbrace$ be an antisymmetric function such that
$f(x,y)=-f(y,x)=1$ and $f(y,x)=-1$ with $f(x,x)=f(y,y)=0$. Given any standard stationary and isotropic kernel
$K(\|x-y\|;\lambda)$ with correlation length $\lambda$, then an antisymmetric kernel is
\begin{align}
\mathlarger{\mathscr{K}}(\|x-y\|;\lambda)=\mathbb{E}[\psi(x)\psi(y)]=f(x,y)K(\|x-y\|;\lambda)
\end{align}
The random field then commute for any pair $(x,y)\in\mathbf{D}$ so that
\begin{align}
\llbracket \psi(x),\psi(y)\rrbracket=\psi(x)\psi(y)-\psi(y)\psi(x)=0
\end{align}
But the expectations do not commute
\begin{align}
\mathbb{E}\big[\llbracket \psi(x),\psi(y)\rrbracket\big]=\mathbb{E}[\psi(x)\psi(y)]-\mathbb{E}[\psi(y)\psi(x)]=2\mathbb{E}[\psi(x)\psi(y)]
\end{align}
\end{prop}
Correlations involving the first derivative $\nabla_{i}\psi(x)$ vanish.
\begin{lem}
\begin{align}
&\mathbb{E}[\psi(x)\nabla_{i}\psi(x)]=0\\&
\mathbb{E}[\nabla_{i}\psi(x)\nabla_{j}\psi(x)]=0
\end{align}
\end{lem}
\begin{proof}
Let $\nabla_{i}^{(x)}=\tfrac{\partial}{\partial x{i}}$ so that $\nabla_{i}^{(x)}\psi(y)=0$. Then
\begin{align}
&\nabla_{i}^{(x)}\mathbb{E}[\psi(x)\psi(y)]=\mathbb{E}[\nabla_{i}^{(x)}\psi(x)\psi(y)]\nonumber\\&=\nabla_{i}f(x,y)K(\|x-y\|;\lambda)
+f(x,y)\nabla_{i}^{(x)}K(\|x-y\|;\lambda)=f(x,y)\nabla_{i}^{(x)}K(\|x-y\|;\lambda)\nonumber\\&
=f(x,y)d\left(\frac{\|x-y\|}{\alpha\lambda^{2}}\right)\left(1-\frac{\|x-y\|^{2}}{2\alpha\lambda^{2}}\right)^{-\alpha-1}
\end{align}
since $\nabla_{i}x=d$. Taking the limit $y\rightarrow x$ gives
\begin{align}
\mathbb{E}[\psi(x)\nabla_{i}\psi(x)]=\lim_{y\rightarrow x}\mathbb{E}[\nabla_{i}^{(x)}\psi(x)\psi(y)]=\lim_{y\rightarrow x}f(x,y)d\left(\frac{\|x-y\|}{\alpha\lambda^{2}}\right)\left(1-\frac{\|x-y\|^{2}}{2\alpha\lambda^{2}}\right)^{-\alpha-1}=0
\end{align}
Taking the derivative again then leads to (2.17).
\end{proof}
\subsection{Random vector fields engineered from random scalar fields}
Given the random (scalar) field $\psi(x)$ and a smooth deterministic spatio-temporal vector field $X_{i}(x,t)$, a new random vector field can
be defined as follows.
\begin{prop}
Let $X_{i}:~\mathbf{D}\times\mathbf{R}^{+}\rightarrow\mathbf{R}^{d}$ be a smooth deterministic vector field existing for all
$(x,t)\in\mathbf{D}\times\mathbf{R}^{+}$. By smooth, the first and second derivatives $\nabla_{j}X_{i}(x,t),\Delta X_{i}(x,t)$ exist, and deterministic
is taken to mean that the field $X_{i}(x,t)$ evolves from initial data $X_{i}(x,0)=X_{i}^{o}(x)$ by some non-linear PDE such that
\begin{align}
\partial_{t}{X}_{i}(x,t)+\mathlarger{\mathscr{D}}_{N}\big[\nabla,\Delta, X_{i}(x,t)\big]X_{i}(x,t)\equiv\partial_{t}{X}_{i}(x,t)+\mathlarger{\mathscr{D}}_{N}X_{i}(x,t)
\end{align}
where $\mathscr{D}_{N}=\mathscr{D}_{N}[\nabla,\Delta, X_{i}]$ is a nonlinear differential operator. A trivial solution is $X_{i}(x,t)=X_{i}$. The Gaussian field $\psi(x)$ has the properties
\begin{align}
&{\mathbb{E}}[\psi(x)]=0\\&
{\mathbb{E}}[\psi(x)\psi(x)]=0\\&
{\mathbb{E}}[\nabla_{i}\psi(x)]=0\\&
{\mathbb{E}}[\Delta\psi(x)]=0\\&
\mathlarger{\mathscr{K}}(\|x-y\|;\lambda)=\bm{\mathbb{E}}[\psi(x)\psi(y)]=f(x,y)K(\|x-y\|;\lambda)
\end{align}
Then a random vector field $\mathcal{X}_{i}(x,t)$ in d dimensions can be defined by a 'mixing' ansatz
\begin{align}
\mathcal{X}_{i}(x,t)=X_{i}(x,t)+\frac{\theta}{\sqrt{d(d+2)}}X_{i}(x,t)\psi(x)
\end{align}
so that the expected value is
\begin{align}
\bm{\mathbb{E}}[\mathcal{X}_{i}(x,t)]=X_{i}(x,t)+\frac{\theta}{\sqrt{d(d+2)}} X_{i}(x,t){\mathbb{E}}[\psi(x)]=X_{i}(x,t)
\end{align}
where $\theta>0$ is a small real parameter.
\end{prop}
The random vector field then satisfies a stochastically averaged nonlinear SPDE. First
\begin{align}
&\partial_{t}\mathcal{X}_{i}(x,t)+\mathlarger{\mathscr{D}}_{N}\mathcal{X}_{i}(x,t)=\partial_{t}\mathcal{X}_{i}(x,t)+\mathlarger{\mathscr{D}}_{N}X_{i}(x,t)
+\mathlarger{\mathscr{D}}_{N}
\left\lbrace\frac{\theta}{\sqrt{d(d+2)}} X_{i}(x,t){\psi}(x)\right\rbrace\nonumber\\&
=\partial_{t}\mathcal{X}_{i}(x,t)+\mathlarger{\mathscr{D}}_{N}
\big\lbrace\frac{\theta}{\sqrt{d(d+2)}}\theta X_{i}(x,t){\psi}(x)\big\rbrace
\end{align}
The stochastically averaged SPDEs is then
\begin{align}
&\mathbb{E}[\mathlarger{\mathscr{D}}_{N}\mathcal{X}_{i}(x,t)]=\mathlarger{\mathscr{D}}_{N}X_{i}(x,t)
+\mathbb{E}\left[\mathlarger{\mathscr{D}}_{N}
\left\lbrace\frac{\theta}{\sqrt{d(d+2)}} X_{i}(x,t){\psi}(x)\right\rbrace\right]\nonumber\\&
=\mathlarger{\mathscr{D}}_{N}\big(\nabla,\Delta\big)
\big\lbrace\theta X_{i}(x,t){\psi}(x)\big\rbrace
\end{align}
New terms may arise in general upon taking the expectation since the underlying PDE is nonlinear.
\subsection{Structure functions}
The structure functions for the vector field $\mathcal{X}_{i}(x,t)$ are now defined. In the Kolmogorov turbulence theory, the third-order structure function in 3 dimensions leads to the famous 4/5 scaling law and the second order structure function to the 2/3-law \textbf{[1-5]}. It will be shown how the same mathematical form arises from the structure functions of the random field $\mathcal{X})_{i}(x,t)$.
\begin{figure}[htb]
\begin{center}
\includegraphics[height=2.9in,width=2.9in]{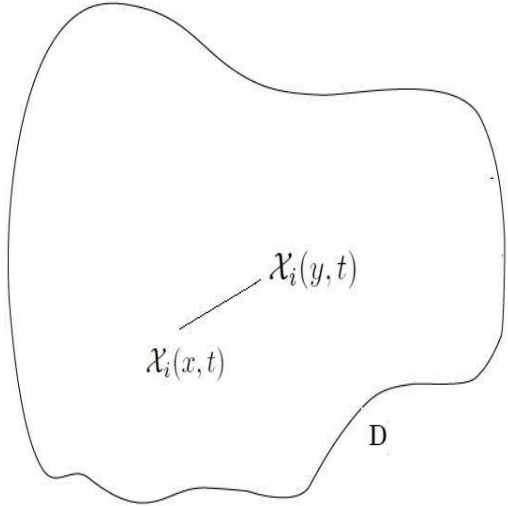}
\caption{Separation between the random fields at two points in $\mathbf{D}$}
\end{center}
\end{figure}The second-order structure function is also equivalent of the square of the canonical metric for the random fields
\begin{defn}
Given the GRVF $\mathcal{X}_{i}(x,t)$ for all $(x,t)\in\mathbf{D}\times \mathbf{R}^{+} $, the \textbf{canonical metric} is defined as \textbf{[58]}
\begin{align}
\mathlarger{d}_{2}(x,y|t)\equiv \mathlarger{d}_{2}(x,y)=\sqrt{\bm{\mathbb{E}}\left[\big|{\mathcal{X}}_{i}(y,t)-{\mathcal{X}}_{i}(x,t)\big|^{2}\right]}
\end{align}
The \textbf{structure functions} ${S}[{\mathcal{X}}]$ of ${\mathcal{X}}_{i}(x,t)$ are then equivalent to the square of the canonical metric
\begin{align}
{S}_{2}[\|x-y\|]\equiv \mathlarger{d}_{2}^{2}(x,y)=\bm{\mathbb{E}}\left[\big|{\mathcal{X}}_{i}(y,t)-{\mathcal{X}}_{i}(x,t)\big|^{2}\right]
\end{align}
If $S_{2}[\mathcal{X}]$ obeys a scaling law over some range of length scales $\ell=|y-x|\le L $ then one expects
\begin{align}
S_{2}[\mathcal{X}]=\bm{\mathbb{E}}\left[\big|{\mathcal{X}}_{i}(y,t)-{\mathcal{X}}_{i}(x,t)\big|^{2}\right]
\sim C|y-x|^{\zeta}=C\ell^{\zeta}
\end{align}
where $\alpha>0$ is some (fractional) power. If $y=x+{\ell}$ with ${\ell}\ll {L}$ and $\mathrm{Vol}(\mathbf{D})\sim {L}^{3}$, then
\begin{align}
{S}_{2}[\ell]\equiv \mathlarger{{d}}_{2}^{2}({x},{x}+{\ell})=\bm{\mathbb{E}}\left[\big|
{\mathcal{X}}_{i}(x+{\ell},t)-{\mathcal{X}}_{i}(x,t)\big|^{2}\right]\nonumber\\
=\bm{\mathbb{E}}\left[\big|{\mathcal{X}}_{i}(x+{\ell},t)\big|^{2}\right]
-2\bm{\mathbb{E}}\left[\big|{\mathcal{X}}_{i}(x+{\ell},t){\mathcal{X}}_{i}(x,t)\big|\right]
+\bm{\mathbb{E}}\left[\big|{\mathcal{X}}_{i}({x},t)\big|^{2}\right]
\end{align}
If $S_{2}[\ell]$ obeys a scaling law then one expects
\begin{align}
S_{2}[\ell]=\bm{\mathbb{E}}\left[\big|{\mathcal{X}}_{i}(x+\ell,t)-{\mathcal{X}}_{i}(x,t)\big|^{2}\right]\sim C\ell^{\zeta}
\end{align}
\end{defn}
\begin{cor}
If $\bm{\ell}=0$ or ${x}=y$ then
\begin{align}
S_{2}[\ell=0]=\mathlarger{{d}}_{2}({x},{x})=0
\end{align}
\end{cor}
It will be convenient always to assume that $S_{2}[{\mathcal{X}}]$ is continuous so that
\begin{align}
&\lim_{y\rightarrow x}\mathlarger{{d}}_{2}^{2}(x,y)=\lim_{y\rightarrow x}\bm{\mathbb{E}}\left[\big|{\mathcal{X}}_{i}(y,t)-{\mathcal{X}}_{i}(x,t)\big|^{2}\right]\nonumber\\&
=\bm{\mathbb{E}}\left[\lim_{y\rightarrow x}\big|{\mathcal{X}}_{i}(y,t)-{\mathcal{X}}_{i}(x,t)\big|^{2}\right]
\end{align}
which is equivalent to the condition
\begin{align}
\bm{\mathbb{P}}\big[\lim_{y\rightarrow x}\big|{\mathcal{X}}_{i}(y,t)-{\mathcal{X}}_{i}({x},t)\big|=0,~\forall~{(x,y)}\in\mathbf{D}\big]=1
\end{align}
which is also consistent with a scaling law.
\subsection{Emergence of a 4/5-law for the 3rd-order structure function}
The main theorem now follows. Beginning with the 'engineered' random field $\mathcal{X}_{i}(x,t)$ it is shown that a 4/5 law emerges when the 3rd-order structure function is computed.
\begin{thm}\textbf{(Emergence of a 4/5-law via 'engineered' random vector fields in $\mathbf{D}\subset\mathbf{R}^{d}$)}\newline
Let a vector field $ X_{i}(x,t)$ evolve within a domain $\mathbf{D}$ of volume $Vol(\mathbf{D})\sim L^{d}$ via some non-linear PDE
\begin{align}
\partial_{t}X_{i}(x,t)+\mathlarger{\mathscr{D}}_{N}X_{i}(x,t)=0,~~(x,t)\in\mathbf{D}\times\mathbf{R}^{+}
\end{align}
and from some initial Cauchy data $X_{i}(x,0)=g_{i}(x)$, where $\mathscr{D}_{N}$ is a nonlinear operator involving $\nabla,\Delta$ and $X(x,t)$ itself.
A trivial steady state solution is then $X_{i}(x,t)=X_{i}=const.$. Let $\psi(x)$ be a random Gaussian scalar field as previously defined, having an antisymmetric rational quadratic covariance kernel
\begin{align}
\mathlarger{\mathscr{K}}(x,y;\lambda)={\mathbb{E}}[\psi(x)\psi(y)]=f(x,y)K(\|x-y\|;\lambda)
\end{align}
Here, $f(x,y)$ is an antisymmetric function $f:\mathbf{D}\times\mathbf{D}\rightarrow \lbrace -1, 0, 1\rbrace $ such that $f(x,y)=-f(y,x)$ with $f(x,y)=1$ for all $(x,y)\in\mathbf{D}$, and $f(y,x)=-1$ with $f(x,x)=f(y,y)=0$. Then $\nabla_{i}^{(x)}f(x,y)=\nabla_{j}^{(y)}f(x,y)=0$. The kernel $K(\|x-y\|;\lambda)$ is any standard stationary  and isotropic covariance kernel for Gaussian random fields; for example a rational quadratic covariance with scale-mixing parameter $\alpha$ gives
\begin{align}
\mathlarger{\mathscr{K}}(x,y;\lambda)={\mathbb{E}}[\psi(x)\psi(y)]= f(x,y)\left(1-\frac{\ell^{2}}{2\alpha\lambda^{2}}\right)^{-\alpha},~~(\alpha,\beta>0)
\end{align}
For a Gaussian quadratic
\begin{align}
\mathlarger{\mathscr{K}}(x,y;\lambda)={\mathbb{E}}[\psi(x)\psi(y)]= f(x,y)\exp\left(-\frac{\|x-y\|^{2}}{2\lambda^{2}}\right)
\end{align}
Then for $y=x+\ell$
\begin{align}
&{\mathbb{E}}[\mathlarger{\psi}(x)]=0\\&
{\mathbb{E}}[\psi(x)\psi(x)]=0\\&
{\mathbb{E}}[\psi(x+\ell)\psi(x+\ell)]=0\\&
{\mathbb{E}}[\psi(x)\psi(x+\ell)]=f(x,x+\ell)K(\|\ell\|;\lambda)\\&
{\mathbb{E}}[\psi(x+\ell)\psi(x)]=-f(x,x+\ell)K(\|\ell\|;\lambda)\equiv f(x+\ell,x)K(\|\ell\|;\lambda) \\&
{\mathbb{E}}[\psi(x)\psi(x)\psi(x)]=0\\&
{\mathbb{E}}[\psi(x+\ell)\psi(x+\ell)\psi(x)]=0\\&
{\mathbb{E}}[\psi(x)\psi(x)\psi(x+\ell)]=0\\&
{\mathbb{E}}[\psi(x+\ell)\psi(x+\ell)\psi(x+\ell)]=0
\end{align}
since all odd moments vanish for a GRF. Now let $\mathbf{Q}=[0,L]$ so that
\begin{align}
\mathbf{Q}=\mathbf{Q}_{1}\bigcup\mathbf{Q}_{2}=[0,\lambda]\bigcup(\lambda,L]\nonumber
\end{align}
then either $\ell\in\mathbf{Q}_{1}$ or $\ell\in\mathbf{Q}_{2}$. We now 'engineer' the following random vector field within $\mathbf{D}\subset\mathbf{R}^{d}$.
\begin{align}
{\mathcal{X}}_{i}(x,t)= X_{i}(x,t)+\frac{\theta}{\sqrt{d(d+2}}X_{i}(x,t)\mathlarger{\psi}(x)
\end{align}
so that $\theta=\frac{1}{\sqrt{15\beta}}$ and $\mathbb{E}[{\mathcal{X}}_{i}(x,t)]=X_{i}(x,t)$. The 3rd-order structure function is then
\begin{align}
\mathlarger{S}_{3}(\ell)={\mathbb{E}}[\left|{\mathcal{X}}_{i}(x+\ell,t)-{\mathcal{X}}_{i}(x,t)\right|^{3}]
\end{align}
Computing $S_{3}(\ell)$ and then letting $X_{i}(x,t)\rightarrow X_{i}=(0,0,X)$, one obtains
\begin{align}
\mathlarger{S}_{3}(\ell)= -\frac{12}{d(d+2)}\theta^{2}\beta\|X_{i}\|^{3}K(\|\ell\|;\lambda)
\end{align}
In three dimensions, $d=3$ and choosing $\theta=1$ gives
\begin{align}
\mathlarger{S}_{3}(\ell)= -\frac{12}{15}\|X_{i}\|^{3}K(\|\ell\|;\lambda)=
-\frac{4}{5}\|X_{i}\|^{3}K(\|\ell\|;\lambda)
\end{align}
Choosing the kernel (2.39) with $\beta=1$
\begin{align}
\mathlarger{S}_{3}(\ell)= -\frac{12}{15}\|X_{i}\|^{3}K(\|\ell\|;\lambda)=
-\frac{4}{5}\|X_{i}\|^{3}\left(1-\frac{\ell^{2}}{2\alpha\lambda^{2}}\right)^{-\alpha}
\end{align}
Then for $\ell\in\mathbf{Q}_{1}=[0,\lambda]$, with $\ell\ll\lambda$, the term $\frac{1}{2}|\ell/\lambda|^{2}$ is very close to zero so that
\begin{align}
\mathlarger{S}_{3}(\ell)= -\frac{4}{5}\|X_{i}\|^{3}
\end{align}
holds over this range of length scales. This is a 4/5-law.
\end{thm}
\begin{proof}
The random fields at $x$ and $x+\ell$ are
\begin{align}
&{\mathcal{X}}_{i}(x,t)=X_{i}(x,t)+\frac{\theta}{\sqrt{d(d+2}}X_{i}(x,t)\mathlarger{\psi}(x)\\&
{\mathcal{X}}_{i}(x+\ell,t)=X_{i}(x+\ell,t)+\frac{\theta}{\sqrt{d(d+2}} X_{i}(x+\ell,t)\mathlarger{\psi}(x+\ell)
\end{align}
Expanding out $\left|{\mathcal{X}}_{i}(x+\ell,t)-{\mathcal{X}}_{i}(x,t)\right|^{3}$
\begin{align}
&\left|{\mathcal{X}}_{i}(x+\ell,t)-{\mathcal{X}}_{i}(x,t)\right|^{3}\nonumber\\&
=\frac{\theta^{3}}{(d(d+2))^{3/2}}|\psi(x+\ell)\psi(x+\ell)\psi(x+\ell)X_{i}(x+\ell)X_{i}(x+\ell)X_{i}(x+\ell)\nonumber\\&
-3\frac{\theta^{3}}{(d(d+2))^{3/2}}X_{i}(x)X(x+\ell)X_{i}(x+\ell)\psi(x)\psi(x+\ell)\psi(x+\ell)\nonumber\\&
+3\frac{\theta^{3}}{(d(d+2))^{3/2}}X_{i}(x)X_{i}(x)X_{i}(x+\ell)\psi(x)\psi(x)\psi(x+\ell)\nonumber\\&
+3\frac{\theta^{2}}{d(d+2)}X_{i}(x)X_{i}(x+\ell)X_{i}(\ell)\psi(x)\psi(x)\nonumber\\&
+3\frac{\theta^{2}}{d(d+2)}X_{i}(x+\ell)X_{i}(x+\ell)X_{i}(x+\ell)\psi(x+\ell)\psi(x+\ell)\nonumber\\&
-3\frac{\theta^{2}}{d(d+2)}X_{i}(x)X_{i}(x+\ell)X_{i}(x+\ell)\psi(x+\ell)\psi(x+\ell)\nonumber\\&
-6\frac{\theta^{2}}{d(d+2)}X_{i}(x)X_{i}(x+\ell)X_{i}(x+\ell)\psi(x)\psi(x+\ell)\nonumber\\&
+6\frac{\theta^{2}}{d(d+2)}X_{i}(x)X_{i}(x)X_{i}(x+\ell)\psi(x+\ell)\psi(x)\nonumber\\&
+3\frac{\theta}{\sqrt{d(d+2)}}X_{i}(x)X_{i}(x+\ell)X_{i}(x+\ell)\psi(x)\nonumber\\&
+6\frac{\theta}{\sqrt{d(d+2)}} X_{i}(x)X_{i}(x)X_{i}(x+\ell)\psi(x)\nonumber\\&
+3\frac{\theta}{\sqrt{d(d+2)}}X_{i}(x+\ell)X_{i}(x+\ell)X_{i}(x+\ell)\psi(x+\ell)\nonumber\\&
+6\frac{\theta}{\sqrt{d(d+2)}} X_{i}(x)X_{i}(x+\ell)X_{i}(x+\ell)\psi(x=\ell)\nonumber\\&
+2\frac{\theta}{\sqrt{d(d+2)}} X_{i}(x)X_{i}(x)X_{i}(x+\ell)\psi(x+\ell)\nonumber\\&
+X_{i}(x+\ell)X_{i}(x+\ell)X_{i}(x+\ell)\nonumber\\&
-3X_{i}(x)X_{i}(x+\ell)X_{i}(x+\ell)\nonumber\\&
+3X_{i}(x)X_{i}(x+\ell)X_{i}(x+\ell)\nonumber\\&
-\frac{\theta^{3}}{(d(d+2))^{3/2}}X_{i}(x)X_{i}(x)X_{i}(x)\psi(x)\psi(x)\psi(x)\nonumber\\&
-3\frac{\theta^{2}}{d(d+2)}X_{i}(x)X_{i}(x)X_{i}(x)\psi(x)\psi(x)\nonumber\\&
-3\frac{\theta}{\sqrt{d(d+2)}} X_{i}(x)X_{i}(x)X_{i}(x)\psi(x)\nonumber\\&
-X_{i}(x)X_{i}(x)X_{i}(x)
\end{align}
Now taking the stochastic expectation and using (2.42-2.50), only the underbraced terms survive
\begin{align}
&\mathlarger{S}_{3}(\ell)=\mathbb{E}[\left|{\mathcal{X}}_{i}(x+\ell,t)-{\mathcal{X}}_{i}(x,t)\right|^{3}]\nonumber\\&
=\frac{\theta^{3}}{(d(d+2))^{3/2}}X_{i}(x+\ell)]X_{i}(x+\ell)X_{i}(x+\ell)\mathbb{E}[|\psi(x+\ell)\psi(x+\ell)\psi(x+\ell)X_{i}(x+\ell)]\nonumber\\&
-3\frac{\theta^{3}}{(d(d+2))^{3/2}}X_{i}(x)X(x+\ell)X_{i}(x+\ell)\mathbb{E}[\psi(x)\psi(x+\ell)\psi(x+\ell)]\nonumber\\&
+3\frac{\theta^{3}}{(d(d+2))^{3/2}}X_{i}(x)X_{i}(x)X_{i}(x+\ell)\mathbb{E}[\psi(x)\psi(x)\psi(x+\ell)]\nonumber\\&
+3\frac{\theta^{2}}{d(d+2)}X_{i}(x)X_{i}(x+\ell)X_{i}(\ell)\mathbb{E}[\psi(x)\psi(x)]\nonumber\\&
+3\frac{\theta^{2}}{d(d+2)}X_{i}(x+\ell)X_{i}(x+\ell)X_{i}(x+\ell)\mathbb{E}[\psi(x+\ell)\psi(x+\ell)]\nonumber\\&
-3\frac{\theta^{2}}{d(d+2)}X_{i}(x)X_{i}(x+\ell)X_{i}(x+\ell)\mathbb{E}[\psi(x+\ell)\psi(x+\ell)]\nonumber\\&
-\underbrace{6\frac{\theta^{2}}{d(d+2)}X_{i}(x)X_{i}(x+\ell)X_{i}(x+\ell)\mathbb{E}[\psi(x)\psi(x+\ell)]}\nonumber\\&
+\underbrace{6\frac{\theta^{2}}{d(d+2)}X_{i}(x)X_{i}(x)X_{i}(x+\ell)\mathbb{E}[\psi(x+\ell)\psi(x)]}\nonumber\\&
+3\frac{\theta}{\sqrt{d(d+2)}}X_{i}(x)X_{i}(x+\ell)X_{i}(x+\ell)\mathbb{E}[\psi(x)]\nonumber\\&
+6\frac{\theta}{\sqrt{d(d+2)}} X_{i}(x)X_{i}(x)X_{i}(x+\ell)\mathbb{E}[\psi(x)]\nonumber\\&
+3\frac{\theta}{\sqrt{d(d+2)}}X_{i}(x+\ell)X_{i}(x+\ell)X_{i}(x+\ell)\mathbb{E}[\psi(x+\ell)]\nonumber\\&
+6\frac{\theta}{\sqrt{d(d+2)}} X_{i}(x)X_{i}(x+\ell)X_{i}(x+\ell)\mathbb{E}[\psi(x+\ell)]\nonumber\\&
+2\frac{\theta}{\sqrt{d(d+2)}} X_{i}(x)X_{i}(x)X_{i}(x+\ell)\mathbb{E}[\psi(x+\ell)]\nonumber\\&
+\underbrace{X_{i}(x+\ell)X_{i}(x+\ell)X_{i}(x+\ell)}\nonumber\\&
-\underbrace{3X_{i}(x)X_{i}(x+\ell)X_{i}(x+\ell)}\nonumber\\&
+\underbrace{3X_{i}(x)X_{i}(x+\ell)X_{i}(x+\ell)}\nonumber\\&
-\frac{\theta^{3}}{(d(d+2))^{3/2}}X_{i}(x)X_{i}(x)X_{i}(x)\mathbb{E}[\psi(x)\psi(x)\psi(x)]\nonumber\\&
-3\frac{\theta^{2}}{d(d+2)}X_{i}(x)X_{i}(x)X_{i}(x)\mathbb{E}[\psi(x)\psi(x)]\nonumber\\&
-3\frac{\theta}{\sqrt{d(d+2)}} X_{i}(x)X_{i}(x)X_{i}(x)\mathbb{E}[\psi(x)]\nonumber\\&
-\underbrace{X_{i}(x)X_{i}(x)X_{i}(x)}
\end{align}
leaving
\begin{align}
&\mathlarger{S}_{3}(\ell)=\mathbb{E}[\left|{\mathcal{X}}_{i}(x+\ell,t)-{\mathcal{X}}_{i}(x,t)\right|^{3}]\nonumber\\&
=-{6\frac{\theta^{2}}{d(d+2)}X_{i}(x)X_{i}(x+\ell)X_{i}(x+\ell)\mathbb{E}[\psi(x)\psi(x+\ell)]}\nonumber\\&
+{6\frac{\theta^{2}}{d(d+2)}X_{i}(x)X_{i}(x)X_{i}(x+\ell)\mathbb{E}[\psi(x+\ell)\psi(x)]}\nonumber\\&
+{X_{i}(x+\ell)X_{i}(x+\ell)X_{i}(x+\ell)}\nonumber\\&
-{3X_{i}(x)X_{i}(x+\ell)X_{i}(x+\ell)}\nonumber\\&
+{3X_{i}(x)X_{i}(x+\ell)X_{i}(x+\ell)}\nonumber\\&
-{X_{i}(x)X_{i}(x)X_{i}(x)}
\end{align}
Now since from (2.45) and (2.46), $\mathbb{E}[\psi(x)\psi(x+\ell)]=-\mathbb{E}[\psi(x+\ell)\psi(x)]$ this becomes
\begin{align}
&\mathlarger{S}_{3}(\ell)=\mathbb{E}[\left|{\mathcal{X}}_{i}(x+\ell,t)-{\mathcal{X}}_{i}(x,t)\right|^{3}]\nonumber\\&
=-{6\frac{\theta^{2}}{d(d+2)}X_{i}(x)X_{i}(x+\ell)X_{i}(x+\ell)\mathbb{E}[\psi(x)\psi(x+\ell)]}\nonumber\\&
-{6\frac{\theta^{2}}{d(d+2)}X_{i}(x)X_{i}(x)X_{i}(x+\ell)\mathbb{E}[\psi(x)\psi(x+\ell)]}\nonumber\\&
+{X_{i}(x+\ell)X_{i}(x+\ell)X_{i}(x+\ell)}\nonumber\\&
-{3X_{i}(x)X_{i}(x+\ell)X_{i}(x+\ell)}\nonumber\\&
+{3X_{i}(x)X_{i}(x+\ell)X_{i}(x+\ell)}\nonumber\\&
-{X_{i}(x)X_{i}(x)X_{i}(x)}
\end{align}
or
\begin{align}
&\mathlarger{S}_{3}(\ell)=\mathbb{E}[\left|{\mathcal{X}}_{i}(x+\ell,t)-{\mathcal{X}}_{i}(x,t)\right|^{3}]\nonumber\\&
=-6\frac{\theta^{2}}{d(d+2)}X_{i}(x)X_{i}(x+\ell)X_{i}(x+\ell)f(x,y)K(\|\ell\|;\lambda)\nonumber\\&
-6\frac{\theta^{2}}{d(d+2)}X_{i}(x)X_{i}(x)X_{i}(x+\ell)f(x,y)K(\|\ell\|;\lambda)\nonumber\\&
+{X_{i}(x+\ell)X_{i}(x+\ell)X_{i}(x+\ell)}\nonumber\\&
-{3X_{i}(x)X_{i}(x+\ell)X_{i}(x+\ell)}\nonumber\\&
+{3X_{i}(x)X_{i}(x+\ell)X_{i}(x+\ell)}\nonumber\\&
-{X_{i}(x)X_{i}(x)X_{i}(x)}
\end{align}
Now let $X_{i}(x,t)=X_{i}(x+\ell,t)=X_{i}=(0,0,X)$ with $\|X_{i}\|=X$, for all $(x,t)\in\mathbf{D}\times\mathbf{R}^{+}$. Then
\begin{align}
&\mathlarger{S}_{3}(\ell)=\mathbb{E}[\left|{\mathcal{X}}_{i}(x+\ell,t)-{\mathcal{X}}_{i}(x,t)\right|^{3}]\nonumber\\&
=-6\frac{\theta^{2}}{d(d+2)}X_{i}X_{i}X_{i}f(x,y)K(\|\ell\|;\lambda)\nonumber\\&
-6\frac{\theta^{2}}{d(d+2)}X_{i}X_{i}X_{i}f(x,y)K(\|\ell\|;\lambda)\nonumber\\&
+X_{i}X_{i}X_{i}-3X_{i}X_{i}X_{i}+3X_{i}X_{i}X_{i}-X_{i}X_{i}X_{i}\nonumber\\&
=-6\frac{\theta^{2}}{d(d+2)}X_{i}X_{i}X_{i}f(x,y)K(\|\ell\|;\lambda)-6\frac{\theta^{2}}{d(d+2)}X_{i}X_{i}X_{i} f(x,y)K(\|\ell\|;\lambda)\nonumber\\&
=-6\frac{\theta^{2}}{d(d+2)}X_{i}X_{i}X_{i}\big(f(x,y)+f(x,y)\big))K(\|\ell\|;\lambda)\nonumber\\&
=-12\frac{\theta^{2}}{d(d+2)}X_{i}X_{i}X_{i}f(x,y)K(\|\ell\|;\lambda)\nonumber\\&
=-12\frac{\theta^{2}}{d(d+2)}X_{i}X_{i}X_{i}K(\|\ell\|;\lambda)\nonumber\\&
=-12\frac{\theta^{2}}{d(d+2)}\|X_{i}\|^{3}K(\|\ell\|;\lambda)
\end{align}
Setting $\theta=1$ since this parameter are arbitrary, and in three dimensions $d=3$ so that
\begin{align}
\mathlarger{S}_{3}(\ell)=-\frac{12\theta^{2}}{15}\|X_{i}\|^{3}K(\|\ell\|;\lambda)=-\frac{4}{5}\|X_{i}\|^{3}K(\|\ell\|;\lambda)
\end{align}
If the rational quadratic covariance (2.6) is chosen then
\begin{align}
\mathlarger{S}_{3}(\ell)=-\frac{4}{5}\|X_{i}\|^{3}K(\|\ell\|;\lambda)=-\frac{4}{5}\|X_{i}\|^{3}\beta\left(1-\frac{|\ell|^{2}}{2\alpha\lambda^{2}}\right)^{-\alpha}
\end{align}
and setting $\beta=1$
\begin{align}
\mathlarger{S}_{3}(\ell)=-\frac{4}{5}\|X_{i}\|^{3}K(\|\ell\|;\lambda)=-\frac{4}{5}\|X_{i}\|^{3}\left(1-\frac{|\ell|^{2}}{2\alpha\lambda^{2}}\right)^{-\alpha}
\end{align}
For the range of length scales for which $\ell/\lambda$ is very small, that is $\ell\in\mathbf{Q}_{1}=[0,\lambda]$ with $\tfrac{\ell^{2}}{2\alpha\lambda^{2}}\ll 1$ then
\begin{align}
\boxed{\mathlarger{S}_{3}(\ell)=-\frac{4}{5}\|X_{i}\|^{3}}
\end{align}
which is a $4/5$ law.
\end{proof}
\begin{rem}
The Gaussian-decaying kernel (2.39) gives the same result since (2.62) then becomes
\begin{align}
\mathlarger{S}_{3}(\ell)=-12\frac{\theta^{2}}{15}\|X_{i}\|^{3}K(\|\ell\|;\lambda)=-\frac{4}{5}\|X_{i}\|^{3}K(\|\ell\|;\lambda)=
-\frac{4}{5}\|X_{i}\|^{3}\beta\exp\left(-\frac{\ell^{2}}{\lambda^{2}}\right)
\end{align}
which is $\mathlarger{S}_{3}(\ell)=-\frac{4}{5}\|X_{i}\|^{3}$ for $\beta=1$ and $\ell/\lambda\sim 0 $ being very small.
\end{rem}
\begin{rem}
Notice that there are two possibilities for the 3rd-order structure function. For $X_{i}(x)=X_{i}(x+\ell)=X_{i}$ one can have
\begin{align}
S_{3}(\ell)=-\frac{6\theta^{2}}{d(d+2)}X_{i}X_{i}X_{i}\mathbb{E}[\psi(x)\psi(x+\ell)]+\frac{6\theta^{2}}{d(d+2)}X_{i}X_{i}X_{i}\mathbb{E}[\psi(x+\ell)\psi(x)]
\end{align}
or
\begin{align}
S_{3}^{*}(\ell)=-\frac{6\theta^{2}}{d(d+2)}X_{i}X_{i}X_{i}\mathbb{E}[\psi(x)\psi(x+\ell)]+\frac{6\theta^{2}}{d(d+2)}X_{i}X_{i}X_{i}\mathbb{E}[\psi(x)\psi(x+\ell)]
\end{align}
since $\psi(x)\psi(x+\ell)\equiv\psi(x+\ell)\psi(x)$ but $\mathbb{E}[\psi(x)\psi(x+\ell)]=-\mathbb{E}[\psi(x+\ell)\psi(x)]$. Then $S_{3}^{*}(\ell)=0$ and
$S_{3}(\ell)=-\tfrac{12}{d(d+2)}X_{i}^{3}\mathbf[\psi(x)\psi(x+\ell)]=-\tfrac{12}{d(d+2)}X_{i}^{3}f(x,y)K(\|x-y\|;\lambda)$. Hence, the non-vanishing option is chosen.
\end{rem}
\subsection{Expression for the 2nd-order structure function}.
Next, the 2nd-order structure function $S_{2}(\ell)$ is computed for the same random field $\mathcal{X}_{i}(x,t)$
\begin{thm}
Let the same scenario of Thm (2.10) hold. Then the 2nd-order structure function is
\begin{align}
\mathlarger{S}_{2}(\ell)=\mathbb{E}[|\mathcal{X}_{i}(x+\ell)-\mathcal{X}_{i}(x)|^{2}]
\end{align}
The random field is again given by (2.48) so that
\begin{align}
\mathcal{X}_{i}(x,t)=X_{i}(x,t)+\frac{\theta}{\sqrt{d(d+2)}}X_{i}(x,t)\psi(x)
\end{align}
with $\psi(x)$ having the same properties and covariance as before. Then
\begin{align}
&|\mathcal{X}_{i}(x+\ell)-\mathcal{X}_{i}(x)|^{2}=\theta^{2}{\theta^{2}}{{d(d+2)}}X_{i}(x+\ell)X_{i}(x+\ell)\psi(x+\ell)\psi(x+\ell)\nonumber\\&
+2\frac{\theta^{2}}{{d(d+2)}}x(X)X(x+\ell)\psi(x+\ell)\psi(x)\nonumber\\&
-2\frac{\theta}{\sqrt{d(d+2)}} X(x)X(x+\ell)\psi(x)\nonumber\\&
-2\frac{\theta}{\sqrt{d(d+2)}}X_{i}(x+\ell)X_{i}(x+\ell)\psi(x+\ell)\nonumber\\&
+2\frac{\theta}{\sqrt{d(d+2)}}X_{i}(x)X_{i}(x+\ell)\psi(x+\ell)\nonumber\\&
+X_{i}(x+\ell)X_{i}(x+\ell)\nonumber\\&
-2X_{i}(x)X_{i}(x+\ell)\nonumber\\&
+\frac{\theta^{2}}{{d(d+2)}}X_{i}(x)X_{i}(x)\nonumber\\&
-X_{i}(x)X_{i}(x)\nonumber\\&
+2\frac{\theta}{\sqrt{d(d+2)}} X_{i}(x)X_{i}(x)\psi(x)
\end{align}
Taking the stochastic expectation, only the underbraced term survives
\begin{align}
&\mathlarger{S}_{2}(\ell)=\mathbb{E}[|\mathcal{X}_{i}(x+\ell)-\mathcal{X}_{i}(x)|^{2}]\nonumber\\&=\theta^{2}{\theta^{2}}{{d(d+2)}}X_{i}(x+\ell)X_{i}(x+\ell)
\mathbb{E}[\psi(x+\ell)\psi(x+\ell)]\nonumber\\&
+\underbrace{2\frac{\theta^{2}}{{d(d+2)}}x(X)X(x+\ell)\mathbb{E}[\psi(x)\psi(x+\ell)]}\nonumber\\&
-2\frac{\theta}{\sqrt{d(d+2)}} X(x)X(x+\ell)\mathbb{E}[\psi(x)]\nonumber\\&
-2\frac{\theta}{\sqrt{d(d+2)}}X_{i}(x+\ell)X_{i}(x+\ell)\mathbb{E}[\psi(x+\ell)]\nonumber\\&
+2\frac{\theta}{\sqrt{d(d+2)}}X_{i}(x)X_{i}(x+\ell)\mathbb{E}[\psi(x+\ell)]\nonumber\\&
+X_{i}(x+\ell)X_{i}(x+\ell)\nonumber\\&
-2X_{i}(x)X_{i}(x+\ell)\nonumber\\&
+\frac{\theta^{2}}{{d(d+2)}}X_{i}(x)X_{i}(x)\nonumber\\&
-X_{i}(x)X_{i}(x)\nonumber\\&
+2\frac{\theta}{\sqrt{d(d+2)}} X_{i}(x)X_{i}(x)\mathbb{E}[\psi(x)]
\end{align}
so that
\begin{align}
&\mathlarger{S}_{2}(\ell)=\mathbb{E}\left[|\mathcal{X}_{i}(x+\ell)-\mathcal{X}_{i}(x)|^{2}\right]=\frac{2\theta^{2}}{d(d+2)}X_{i}(x)X_{i}(x+\ell)
\mathbb{E}[\psi(x)\psi(x+\ell)]\nonumber\\&
=\frac{2\theta^{2}}{d(d+2)}X_{i}(x)X_{i}(x+\ell)f(x,x+\ell)K(\|\ell\|;\lambda)
\end{align}
Setting $X_{i}(x+\ell)=X_{i}(x)=X_{i}=(0,0,X)$
\begin{align}
&\mathlarger{S}_{2}(\ell)=\frac{2\theta^{2}}{d(d+2)}X_{i}X_{i}f(x,x+\ell)K(\|\ell\|;\lambda)
\nonumber\\&\equiv
\frac{2\theta^{2}}{d(d+2)}\|X_{i}\|^{2}f(x,x+\ell)K(\|\ell\|;\lambda)\nonumber\\&
\equiv\frac{2\theta^{2}}{d(d+2)}\|X_{i}\|^{2}K(\|\ell\|;\lambda)
\end{align}
For $d=3$
\begin{align}
\mathlarger{S}_{2}(\ell)=\frac{2\theta^{2}}{15}\|X_{i}\|^{2}K(\|\ell\|;\lambda)=C\|X_{i}\|^{2}K(\|\ell\|;\lambda)
\end{align}
Choosing the kernel (2.39)
\begin{align}
\mathlarger{S}_{2}(\ell)=C\|X_{i}\|^{2}\left(1+\frac{\ell^{2}}{2\alpha\lambda^{2}}\right)^{-\alpha}
\end{align}
When $\ell\in\mathbf{Q}=[0,\lambda]$ and $\ell\ll\lambda$ then $(\ell^{2}/\lambda^{2})\sim 0$.
\begin{align}
\boxed{\mathlarger{S}_{2}(\ell)=C\|X_{i}\|^{2}}
\end{align}
\end{thm}
\section{APPLICATION TO FLUID MECHANICS AND THE EMERGENCE OF THE CLASSICAL KOLMOGOROV SCALING LAWS OF TURBULENCE}
It can now be demonstrated  that the mathematical form of the classical Kolmogorov scaling laws for turbulence
emerge from the formalism when the random field $\mathcal{X}_{i}(x,t)$ is identified with a turbulent velocity field for $d=3$; that is, the emergence of both the \textbf{K41} 2/3-law and the \textbf{K41} 4/5-law.
\subsection{Basic results from fluid mechanics}
Some basic background results from smooth or deterministic 'laminar' fluid mechanics are briefly given \textbf{[33,34,50,51,54]}. In the absence of turbulence, we consider a set of smooth  and deterministic solutions $(U_{i}(x,t),\rho)$ of the steady state viscous Burgers equations, with the pressure gradient term set to zero in the Navier-Stokes equations. Here $U_{i}(x,t)$ is the steady state fluid velocity at $x\in\mathbf{D}\subset{\mathbf{R}}^{d}$, and $\rho$ is the (uniform) density. For the general Burgers equations, let $\mathbf{D}\subset\bm{\mathbf{R}}^{3}$ be a compact bounded domain with ${x}\in\mathbf{D}$ and filled with a fluid of density $\rho:[0,T]\times\mathbf{R}^{3}\rightarrow{\mathbf{R}}_{\ge}$, and velocity $U:\mathbf{D}\times\mathbf{R}^{(+)}\rightarrow{\mathbf{R}}^{3}$ with $U_{i}(x,t))_{1\le i\le d}$ so that $\rho=\rho(x,t)$ and $U_{i}=U_{i}(x,t)$.
The Burger's equations is then
\begin{align}
&\partial_{t}U_{i}(x,t)+\mathlarger{\mathscr{D}}_{N}U_{i}(x,t)\nonumber\\&\equiv\partial_{t}U_{i}(x,t)
-\nu\Delta U_{i}(x,t)+U^{j}(x,t)\nabla_{j}U_{i}(x,t)=0,~(x,t)\in\mathbf{D}\times\mathbf{R}^{+}
\end{align}
The viscosity of the fluid is $\nu$ is very small so that $\nu\sim 0$, with the incompressibility condition $\nabla_{i}U^{i}=0$.
\begin{defn}
The following will also apply:
\begin{enumerate}[(a)]
\item The smooth initial Cauchy data is $U(x,0)=U_{o}(x)$. One could also impose periodic boundary conditions if $\mathbf{D}$ is a cube or box of with sides of length $\mathrm{L}$ such that $U_{i}({x}+{L},t)=U_{i}(x,t)$, or no-slip BCs $U_{i}(x,t)=0, \forall~x\in\partial\mathbf{D}$. For some $C,K>0$, the initial data will also satisfy a bound of the typical form $ |\nabla_{x}^{\alpha}U_{o}(x)|\le C(\alpha,K)(1+|x|)^{-K}$. By a \textbf{\textit{smooth deterministic flow}}, we mean a $U_{i}(x,t)$ which is deterministic and non-random  and evolves \textit{predictably} by the NS equations from some initial Cauchy data $ U_{i}(x,0)=U_{o}(x)$. For example, a simple trivial laminar flow solution ith $U_{i}(x,t)=U_{i}
=const$. A generic smooth flow will be differentiable to at least 2nd order so that $\nabla_{j}U_{i}(x,t)$ and $\nabla_{i}\nabla_{j}U_{i}(x,t)$ exist. The fluid velocity $U_{i}(x,t)$ is a divergence-free vector field that should be physically reasonable: that is, the solution should not
grow too large or blow up as $t\rightarrow \infty$.
\item The Reynolds number within $\mathbf{D}$ with $Vol(\mathbf{D})\sim L^{d}$ is
\begin{align}
\mathscr{R}(x,t;\nu,L)=\frac{\|U_{i}(x,t)\| L}{\nu}
\end{align}
and for a constant velocity $U_{i}(x,t)=U_{i}$
\begin{align}
\mathscr{R}=\frac{\|U_{i}\| L}{\nu}
\end{align}
For $\nu>0$ but $\nu \sim 0$, then the Reynolds number will be very large but not infinite.
\item Given the flow $U_{i}(x,t)$ and a closed curve or knot $\Gamma\in\mathbf{D}$ with $x\in\Gamma$ then the circulation is
\begin{align}
C(\Gamma)=\oint_{\Gamma}U_{i}(x,t)dx^{i}
\end{align}
and the vorticity is
\begin{align}
W^{i}(x,t)=\varepsilon^{ijk}\nabla_{j}U_{k}(x,t)
\end{align}
\item The basic energy balance equation for a viscous fluid is obeyed such that
\begin{align}
\|U_{i}(\bullet,t\|_{L_{2}(\mathbf{D})}^{2}=\frac{d E(t)}{dt}=\frac{d}{dt}\mathlarger{\int}_{\mathbf{D}}U_{i}(x,t)U^{i}(x,t)d^{3}x= -
\nu\mathlarger{\int}_{\mathbf{D}}|\nabla_{j}U_{i}(x,t)\nabla^{j}U^{i}(x,t)|^{2}d^{3}x
\end{align}
or $\frac{d E(t)}{dt}=-\nu\mathfrak{E}(t)$, where $\mathfrak{E}$ is the enstrophy.
\item The energy dissipation rate for a constant viscosity $\nu$ is
\begin{align}
\mathlarger{\epsilon}(x,t)=\frac{1}{2}\nu\left(\nabla_{j}U_{i}(x,t)+\nabla^{i}U^{j}(x,t)\right)^{2}
\end{align}
If the fluid is isotropic then $\nabla_{j}U_{i}=\nabla_{i}U_{j}$ so that
\begin{align}
\mathlarger{\epsilon}(x,t)=\nu\left(\nabla_{j}U_{i}(x,t)\nabla^{j}U^{i}(x,t)\right)
\end{align}
Then the volume integral is
\begin{align}
\int_{\mathbf{D}}\mathlarger{\epsilon}(x,t)d^{3}x=\nu\int_{\mathbf{D}}\left(\nabla_{j}U_{i}(x,t)\nabla^{j}U^{i}(x,t)\right)^{2}d^{3}x
\end{align}
which is minus the rhs of (3.4) and is the enstrophy integral. The energy balance equation is then
\begin{align}
\frac{d}{dt}\mathlarger{\int}_{\mathbf{D}}U_{i}(x,t)U^{i}(x,t)d^{3}x=-\int_{\mathbf{D}}\mathlarger{\epsilon}(x,t)d^{3}x
\end{align}
or
\begin{align}
\left|\frac{d}{dt}\mathlarger{\int}_{\mathbf{D}}U_{i}(x,t)U^{i}(x,t)d^{3}x\right|=\left|-\int_{\mathbf{D}}\mathlarger{\epsilon}(x,t)d^{3}x\right|\equiv
\left|\int_{\mathbf{D}}\mathlarger{\epsilon}(x,t)d^{3}x\right|
\end{align}
\end{enumerate}
\end{defn}
\begin{rem}
From (3.6), one might naturally expect that in the limit as $\nu\rightarrow 0$ then the energy dissipation rate $\epsilon(x,t)$ simply vanishes. However, this
(naive) expected outcome is contradicted by observation, experimentally and numerically--the dissipation rate always remains finite as $\nu$ vanishes, and is independent of $\nu$. This is known as anomalous dissipation (AD) in the fluid mechanics literature. This can also be expressed as
\begin{align}
\lim_{\nu\rightarrow 0} \nu|\nabla_{j}U_{i}(x,t)|^{2}=\epsilon>0
\end{align}
Kolmogorov assumed AD in his original 1941 derivations of the 2/3 and 4/5 laws. The reason for AD is that no matter how small $\nu$ is there is always a cascade of energy from larger to smaller scales. Onsanger interpreted AD in terms of the Holder continuity of weak solutions of the Euler equations \textbf{[13]} but there is no rigourous mathematical description of AD.
\end{rem}
We now examine an 'engineered' random field $\mathcal{U}_{i}(x,t)$ representing a turbulent fluid flow, and show that the form of the 4/5 law emerges when the 3rd-order structure function is computed.
\begin{prop}
The random field $\mathcal{U}_{i}(x,t)$ has the same form as (2.24) so that for all $x\in\mathbf{D}\subset\mathbf{R}^{d},t\in\mathbf{R}^{+}$.
\begin{align}
{\mathcal{U}}_{i}(x,t)=U_{i}(x,t)+\frac{\theta}{\sqrt{d(d+2}}U_{i}(x,t)\mathlarger{\psi}(x)
\end{align}
and when $U_{i}(x,t)=U_{i}$
\begin{align}
{\mathcal{U}}_{i}(x,t)=U_{i}+\frac{\theta}{\sqrt{d(d+2}}U_{i}\mathlarger{\psi}(x)
\end{align}
where $U_{i}(x,t)$ evolves by the Burgers equation and $U_{i}$ is a trivial steady state solution.
In three dimensions with $d=3$
\begin{align}
&{\mathcal{U}}_{i}(x,t)=U_{i}(x,t)+\frac{\theta}{\sqrt{15}}U_{i}(x,t)\mathlarger{\psi}(x)\\&
{\mathcal{U}}_{i}(x,t)=U_{i}+\frac{\theta}{\sqrt{15}}U_{i}\mathlarger{\psi}(x)
\end{align}
The expectation gives the mean flow velocity
\begin{align}
\mathbb{E}[{\mathcal{U}}_{i}(x,t)]=U_{i}(x,t)+\frac{\theta}{\sqrt{d(d+2}}U_{i}(x,t)\mathbb{E}[\mathlarger{\psi}(x)]=U_{i}(x,t)
\end{align}
and when $U_{i}(x,t)=U_{i}$
\begin{align}
\mathbb{E}[{\mathcal{U}}_{i}(x,t)]=U_{i}+\frac{\theta}{\sqrt{d(d+2}}U_{i}\mathbb{E}[\mathlarger{\psi}(x)]=U_{i}
\end{align}
\end{prop}
It can be briefly demonstrated that the turbulent flow or random field $\mathcal{U}_{i}(x,t)$ is compatible with basic important aspects of fluid mechanics.
\begin{lem}
The flow $U_{i}(x,t)$ is isotropic if $\nabla_{j}U_{i}(x,t)=\nabla_{i}U_{j}(x,t)$ and incompressible if $\nabla^{j}U_{j}(x,t)=0$. It then follows that for the turbulent flow
\begin{align}
&\mathbb{E}[\nabla^{j}U_{j}(x,t)]=0\\&
\mathbb{E}[\nabla_{i}\mathcal{U}_{j}(x,t)]=\mathbb{E}[\nabla_{j}\mathcal{U}_{i}(x,t)]
\end{align}
\end{lem}
\begin{proof}
\begin{align}
\nabla^{j}U_{j}(x,t)=\nabla^{j}U_{j}(x,t)+\frac{\theta}{\sqrt{d(d+2)}}\nabla_{j}U_{j}(x,t)\psi(x)+\frac{\theta}{\sqrt{d(d+2)}}U_{j}(x,t)\nabla_{j}\psi(x)
\end{align}
The expectation is then
\begin{align}
&\mathbb{E}[\nabla^{j}\mathcal{U}_{j}(x,t)]=\nabla^{j}U_{j}(x,t)+\frac{\theta}{\sqrt{d(d+2)}}\nabla_{j}U_{j}(x,t)\mathbb{E}[\psi(x)]\nonumber\\&
+\frac{\theta}{\sqrt{d(d+2)}}U_{j}(x,t)
\mathbb{E}[\nabla_{j}\psi(x)]=\nabla^{j}U_{j}(x,t)
\end{align}
and
\begin{align}
&\mathbb{E}[\nabla^{j}\mathcal{U}_{i}(x,t)]=\nabla^{j}U_{i}(x,t)+\frac{\theta}{\sqrt{d(d+2)}}\nabla_{j}U_{i}(x,t)\mathbb{E}[\psi(x)\\&]+\frac{\theta}{\sqrt{d(d+2)}}U_{i}(x,t)
\mathbb{E}[\nabla_{j}\psi(x)]=\nabla^{j}U_{i}(x,t)\\&
\mathbb{E}[\nabla^{i}\mathcal{U}_{j}(x,t)]=\nabla^{i}U_{j}(x,t)+\frac{\theta}{\sqrt{d(d+2)}}\nabla_{i}U_{j}(x,t)\mathbb{E}[\psi(x)]\\&+\frac{\theta}{\sqrt{d(d+2)}}U_{j}(x,t)
\mathbb{E}[\nabla_{j}\psi(x)]=\nabla^{i}U_{j}(x,t)
\end{align}
\end{proof}
\begin{lem}
For the turbulent flow $\mathcal{U}_{i}(x,t)$ with the Gaussian field $\psi(x)$ having the properties
\begin{align}
&{\mathbb{E}}[\psi(x)]=0\\&
{\mathbb{E}}[\psi(x)\psi(x)]=0\\&
{\mathbb{E}}[\nabla_{i}\psi(x)]=0\\&
{\mathbb{E}}[\nabla_{i}\psi(x)\nabla^{i}\psi(x)]\\&
{\mathbb{E}}[\Delta\psi(x)]=0\\&
\mathlarger{\mathscr{K}}(\|x-y\|;\lambda)=\bm{\mathbb{E}}[\psi(x)\psi(y)]=f(x,y)K(\|x-y\|;\lambda)
\end{align}
\begin{enumerate}
\item The averaged energy integral is
\begin{align}
\mathbb{E}\left[\int_{\mathbf{D}}\mathcal{U}_{i}(x,t)\mathcal{U}^{i}(x,t)d^{3}x\right]=\int_{\mathbf{D}}U_{i}(x,t)U^{i}(x,t)d^{3}x
\end{align}
\item The averaged energy dissipation is
\begin{align}
\mathbb{E}[\mathcal{E}(x,t)]=\nu\mathbb{E}[|\nabla_{j}\mathcal{U}_{i}\nabla^{j}\mathcal{U}^{i}|]=\nu\nabla_{j}{U}_{i}\nabla^{j}{U}^{i}=\epsilon(x,t)
\end{align}
\item The averaged energy balance law for a fixed viscosity $\nu$ is then
\begin{align}
\frac{d}{dt}\mathbb{E}\left[\int_{\mathbf{D}}\mathcal{U}_{i}(x,t)\mathcal{U}^{i}(x,t)d^{3}x\right]=-\nu\int_{\mathbf{D}}\mathbb{E}[|\nabla_{j}\mathcal{U}_{i}\nabla^{j}
\mathcal{U}^{i}|]d^{3}x
\end{align}
iff
\begin{align}
\frac{d}{dt}\int_{\mathbf{D}}{U}_{i}(x,t){U}^{i}(x,t)d^{3}x=-\nu\int_{\mathbf{D}}[|\nabla_{j}{U}_{i}\nabla^{j}
{U}^{i}|]d^{3}x
\end{align}
which is the basic energy balance law for an isotropic fluid with $\nabla_{i}U_{j}(x,t)=\nabla_{j}U_{i}(x,t)$.
\end{enumerate}
\end{lem}
\begin{proof}
To prove (1)
\begin{align}
&\mathbb{E}\left[\int_{\mathbf{D}}\mathcal{U}_{i}(x,t)\mathcal{U}^{i}(x,t)d^{3}x\right]\nonumber\\&=\mathlarger{\int}_{\mathbf{D}}U_{i}(x)U^{i}(x)d^{3}x
+\frac{2\theta}{\sqrt{d(d+2)}}\mathlarger{\int}_{\mathbf{D}}U_{i}(x,t)U^{i}(x,t)\mathbb{E}\left[\psi(x)\right]d^{3}x\\&
+\frac{\theta^{2}}{{d(d+2)}}\mathlarger{\int}_{\mathbf{D}}U_{i}(x,t)U^{i}(x,t)\mathbb{E}\left[\psi(x)\psi(x)\right]d^{3}x
=\mathlarger{\int}_{\mathbf{D}}U_{i}(x)U^{i}(x)d^{3}x
\end{align}
To prove (2)
\begin{align}
&\mathbb{E}\left[\int_{\mathbf{D}}\nabla_{j}\mathcal{U}_{i}(x,t)\nabla^{j}\mathcal{U}^{i}(x,t)d^{3}x\right]\nonumber\\&
=\mathlarger{\int}_{\mathbf{D}}\nabla_{j}U_{i}(x)\nabla^{j}U^{i}(x)d^{3}x
+\frac{2\theta}{\sqrt{d(d+2)}}\mathlarger{\int}_{\mathbf{D}}\nabla_{j}U_{i}(x,t)\nabla^{j}U^{i}(x,t)\mathbb{E}\left[\psi(x)\right]d^{3}x\\&
+\frac{\theta^{2}}{{d(d+2)}}\mathlarger{\int}_{\mathbf{D}}\nabla_{j}U_{i}(x,t)\nabla^{j}U^{i}(x,t)
\mathbb{E}\left[\nabla_{j}\psi(x)\nabla^{j}\psi(x)\right]d^{3}x=\mathlarger{\int}_{\mathbf{D}}\nabla_{j}U_{i}(x)\nabla^{j}U^{i}(x)d^{3}x
\end{align}
then (3.24) holds if (3.25) holds. Note that the Fubuni theorem has been applied which states that if $\mathcal{Y}_{i}(x,t)$ is a generic random field then
$\mathbb{E}[\int\mathcal{Y}_{i}(x,t)d^{3}x]\equiv\int\mathbb{E}[\mathcal{Y}_{i}(x,t)]d^{3}x$.
\end{proof}
The presence of eddies and vortices over a very large range of length and times scales is a characteristic feature of turbulence \textbf{[57]}. The extent of this range is essentially determined by the Reynolds number, which is the ratio of inertial to viscous forces. For very large but not infinite Reynolds numbers ($\nu\sim 0$) the inertial term or nonlinearity dominates resulting in eddies/vortices of many scales being created and so the flow is turbulent; for example in the ocean and atmosphere, vortices can range from hundreds of kilometers to a few millimetres. In the Kolmogorov theory, energy is transferred from large eddies to smaller ones down to the Kolmogorov length scale $\eta$ and dissipated as heat. But as the viscosity is increased turbulence tends to be suppressed; for example, turbulence is very highly suppressed or eliminated in flowing honey or olive oil. The following lemma shows that the random field $\mathcal{U}_{i}(x,t)$ can lead to vortex 'tanging' or correlations.
\begin{figure}[htb]
\begin{center}
\includegraphics[height=2.4in,width=2.4in]{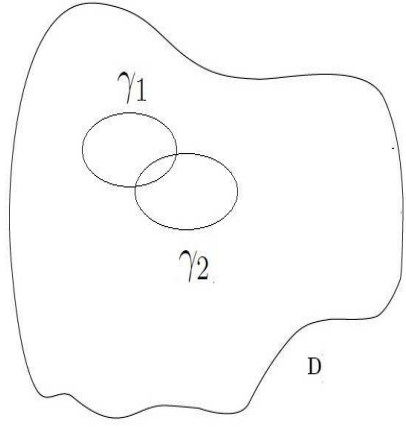}
\caption{Correlation or 'tangle' of two vortices $\gamma_{1}$ and $\gamma_{2}$ within $\mathbf{D}$}
\end{center}
\end{figure}
\begin{lem}
As before, the Gaussian field $\psi(x)$ has the properties
\begin{align}
&{\mathbb{E}}[\psi(x)]=0\\&
{\mathbb{E}}[\psi(x)\psi(x)]=0\\&
{\mathbb{E}}[\nabla_{i}\psi(x)]=0\\&
{\mathbb{E}}[\nabla_{i}\psi(x)\nabla^{i}\psi(x)]\\&
{\mathbb{E}}[\Delta\psi(x)]=0\\&
\mathlarger{\mathscr{K}}(\|x-y\|;\lambda)=\bm{\mathbb{E}}[\psi(x)\psi(y)]=f(x,y)K(\|x-y\|;\lambda)
\end{align}
Take $d=3$ and let $\gamma_{1},\gamma_{2}\in\mathbf{D}\subset\mathbf{R}^{3}$ be curves or knots in $\mathbf{D}$ with $x\in\gamma_{1},y\in\gamma_{2}$. The circulations at $(x,y)$ are then
\begin{align}
\mathrm{C}(\gamma_{1})=\int_{\gamma_{1}}U_{i}(x,t)dx^{i},~~~~\mathrm{C}(\gamma_{2})=\int_{\gamma_{2}}U_{j}(y,t)dy^{i}
\end{align}
For closed curves or knots
\begin{align}
\mathrm{C}(\gamma_{1})=\oint_{\gamma_{1}}U_{i}(x,t)dx^{i},~~~~\mathrm{C}(\gamma_{2})=\oint_{\gamma_{2}}U_{j}(y,t)dy^{i}
\end{align}
The turbulent flows or random fields at (x,y) are
\begin{align}
&\mathcal{X}_{i}(x,t)=X_{i}(x,t)+\frac{\theta}{\sqrt{15}}U_{i}(x,t)\psi(x)\\&
\mathcal{X}_{j}(y,t)=X_{j}(y,t)+\frac{\theta}{\sqrt{15}}U_{j}(y,t)\psi(y)
\end{align}
Stochastic circulations are then proposed as
\begin{align}
&\mathlarger{\mathcal{C}}(\gamma_{1})=\oint_{\gamma_{1}}\mathcal{X}_{i}(x,t)dx^{i}=\oint_{\gamma_{1}}X_{i}(x,t)dx^{i}
+\frac{\theta}{\sqrt{15}}\oint_{\gamma_{1}}U_{i}(x,t)\psi(x)d^{3}x\\&\mathlarger{\mathcal{C}}(\gamma_{2})=\oint_{\gamma_{2}}\mathcal{X}_{j}(y,t)dx^{i}=\oint_{\gamma_{1}}X_{j}(y,t)dy^{j}
+\frac{\theta}{\sqrt{15}}\oint_{\gamma_{1}}U_{j}(x,t)\psi(x)d^{3}x
\end{align}
with expectations
\begin{align}
&\mathbb{E}[\mathlarger{\mathcal{C}}(\gamma_{1})]=\mathbb{E}\left[\oint_{\gamma_{1}}\mathcal{X}_{i}(x,t)dx^{i}\right]=\oint_{\gamma_{1}}X_{i}(x,t)dx^{i}
+\frac{\theta}{\sqrt{15}}\oint_{\gamma_{1}}U_{i}(x,t)\mathbb{E}[\psi(x)]dx^{i}=C(\gamma_{1})\\&
\mathbb{E}[\mathlarger{\mathcal{C}}(\gamma_{2})]=\mathbb{E}\left[\oint_{\gamma_{2}}\mathcal{X}_{j}(y,t)dy^{i}\right]=\oint_{\gamma_{2}}X_{j}(y,t)dy^{i}
+\frac{\theta}{\sqrt{15}}\oint_{\gamma_{2}}U_{j}(y,t)\mathbb{E}[\psi(y)]dy^{j}=C(\gamma_{2})
\end{align}
Now define a correlation or 'tangle' between 2 vortices as $\mathlarger{\mathcal{C}}(\gamma_{1})\bigcap \mathlarger{\mathcal{C}}(\gamma_{2})$. Then
\begin{align}
\mathbb{E}\left[\mathlarger{\mathcal{C}}(\gamma_{1})\bigcap \mathlarger{\mathcal{C}}(\gamma_{2})\right]=\oint\!\oint_{\gamma_{1},\gamma_{2}}
\mathbb{E}[\mathcal{U}_{i}(x,t)\mathcal{U}(y,t)]dx^{i}dy^{j}
\end{align}
which is
\begin{align}
&\mathbb{E}\left[\mathlarger{\mathcal{C}}(\gamma_{1})\bigcap \mathlarger{\mathcal{C}}(\gamma_{2})\right]=\oint\!\oint_{\gamma_{1},\gamma_{2}}
{U}_{i}(x,t){U}(y,t)dx^{i}dy^{j}\nonumber\\&+\oint\!\oint_{\gamma_{1},\gamma_{2}}f(x,y)K(\|x-y\|;\lambda)U_{i}(x,t)U_{j}(y,t)dx^{i}dy^{j}
\end{align}
Then
\begin{align}
\mathbb{E}\left[\mathlarger{\mathcal{C}}(\gamma_{1})\bigcap \mathlarger{\mathcal{C}}(\gamma_{2})\right]\longrightarrow\oint\!\oint_{\gamma_{1},\gamma_{2}}
{U}_{i}(x,t){U}(y,t)dx^{i}dy^{j}
\end{align}
when $\|x-y\|\gg \lambda$, and the vortices are uncorrelated or 'untangled'.
\end{lem}
\begin{proof}
The proof follows from
\begin{align}
\mathbb{E}[\mathcal{U}_{i}(x,t)\mathcal{U}_{j}(y,t)]=U_{i}(x,t)U_{j}(y,t)+f(x,y)K(\|x-y\|;\lambda)U_{i}(x,t)U_{j}(y,t)
\end{align}
\end{proof}
\begin{lem}
The turbulent flow $\mathcal{U}_{i}(x,t)$ is a solution of the stochastically averaged Burger's equation. As before, the Gaussian field $\psi(x)$ has the properties
\begin{align}
&{\mathbb{E}}[\psi(x)]=0\\&
{\mathbb{E}}[\psi(x)\psi(x)]=0\\&
{\mathbb{E}}[\nabla_{i}\psi(x)]=0\\&
\mathbb{E}[\nabla_{i}\psi(x)\psi_{i}(x)]=0\\&
{\mathbb{E}}[\nabla_{i}\psi(x)\nabla^{i}\psi(x)]\\&
{\mathbb{E}}[\Delta\psi(x)]=0\\&
\mathlarger{\mathscr{K}}(\|x-y\|;\lambda)={\mathbb{E}}[\psi(x)\psi(y)]=f(x,y)K(\|x-y\|;\lambda)
\end{align}
Then
\begin{align}
&\mathbb{E}[\partial_{t}\mathcal{U}_{i}(x,t)-\nu\Delta \mathcal{U}_{i}(x,t)+\mathcal{U}^{j}(x,t)\nabla_{j}\mathcal{U}_{i}(x,t)]\nonumber\\&
=\partial_{t}{U}_{i}(x,t)-\nu\Delta {U}_{i}(x,t)+{U}^{j}(x,t)\nabla_{j}{U}_{i}(x,t)=0
\end{align}
\end{lem}
\begin{proof}
First
\begin{align}
&\partial_{t}\mathcal{U}_{i}(x,t)-\nu\Delta \mathcal{U}_{i}(x,t)+\mathcal{U}^{j}(x,t)\nabla_{j}\mathcal{U}_{i}(x,t)\\&
=\partial_{t}U_{i}(x,t)+\partial_{t}U_{i}(x,t)\left(\frac{\theta}{\sqrt{d(d+2)}}\psi(x)\right)\nonumber\\&
-\nu\Delta U_{i}(x,t)-\frac{\nu\theta}{\sqrt{d(d+2)}}\left(\psi(x)\Delta U_{i}(x,t)-\nabla_{i}U_{i}(x,t)\nabla^{j}\psi(x)\right)\nonumber\\&
-\frac{\nu\theta}{\sqrt{d(d+2)}}\left(\nabla^{j}U_{i}\nabla_{j}\psi(x)+U_{i}(x,t)\Delta\psi(x)\right)\nonumber\\&
+U^{j}(x,t)\nabla_{j}U_{i}(x,t)+\frac{\theta}{\sqrt{d(d+2)}}U^{j}(x,t)\psi(x)\nabla_{j}U_{i}(x,t)\nonumber\\&+\frac{\theta}{\sqrt{d(d+2)}}U^{j}(x,t)U_{i}(x,t)\nabla_{j}\psi(x)\nonumber\\&
\frac{\theta}{\sqrt{d(d+2)}}U^{j}(x,t)\nabla_{j}U_{i}(x,t)\psi(x)+\frac{\theta^{2}}{{d(d+2)}}U^{j}(x,t)\nabla_{j}U_{i}(x,t)\psi(x)\psi(x)\nonumber\\&
+\frac{\theta^{2}}{{d(d+2)}}U^{j}(x,t)\nabla_{j}U_{i}(x,t)\nabla_{j}\psi(x)
\end{align}
Taking the stochastic expectation only the underbraced terms survive
\begin{align}
&\mathbb{E}[\partial_{t}\mathcal{U}_{i}(x,t)-\nu\Delta \mathcal{U}_{i}(x,t)+\mathcal{U}^{j}(x,t)\nabla_{j}\mathcal{U}_{i}(x,t)]\\&
=\underbrace{\partial_{t}U_{i}(x,t)}+\partial_{t}U_{i}(x,t)\left(\frac{\theta}{\sqrt{d(d+2)}}\mathbb{E}[\psi(x)]\right)\nonumber\\&
-\underbrace{\nu\Delta U_{i}(x,t)}-\frac{\nu\theta}{\sqrt{d(d+2)}}\left(\mathbb{E}[\psi(x)]\Delta U_{i}(x,t)-\nabla_{i}U_{i}(x,t)\mathbb{E}[\nabla^{j}\psi(x)]\right)\nonumber\\&
-\frac{\nu\theta}{\sqrt{d(d+2)}}\left(\nabla^{j}U_{i}\mathbb{E}[\nabla_{j}\psi(x)]+U_{i}(x,t)\mathbb{E}[\Delta\psi(x)]\right)\nonumber\\&
+\underbrace{U^{j}(x,t)\nabla_{j}U_{i}(x,t)}+\frac{\theta}{\sqrt{d(d+2)}}U^{j}(x,t)\psi(x)\nabla_{j}U_{i}(x,t)\nonumber\\&+\frac{\theta}{\sqrt{d(d+2)}}U^{j}(x,t)U_{i}(x,t)
\mathbb{E}[\nabla_{j}\psi(x)]\nonumber\\&
\frac{\theta}{\sqrt{d(d+2)}}U^{j}(x,t)\nabla_{j}U_{i}(x,t)\mathbb{E}[\psi(x)]+\frac{\theta^{2}}{{d(d+2)}}U^{j}(x,t)\nabla_{j}U_{i}(x,t)\mathbb{E}[\psi(x)\psi(x)]\nonumber\\&
+\frac{\theta^{2}}{{d(d+2)}}U^{j}(x,t)\nabla_{j}U_{i}(x,t)\mathbb{E}[\nabla_{j}\psi(x)]\nonumber\\&
=\partial_{t}{U}_{i}(x,t)-\nu\Delta {U}_{i}(x,t)+{U}^{j}(x,t)\nabla_{j}{U}_{i}(x,t)=0
\end{align}
\end{proof}
\subsection{Emergence of the classical 4/5 and 2/3 scaling laws}
Before reproducing the classical 4/5 and 2/3 scaling laws, the following preliminary lemma is necessary.
\begin{lem}\textbf{Energy dissipation rate and constant velocity}\newline
For a constant velocity $U_{i}(x,t)=U_{i}=(0,0,U)$ and $U=const.$ The energy dissipation rate $\epsilon$ is
\begin{align}
\mathlarger{\mathlarger{\epsilon}}t=|U_{i}U^{i}|\equiv \|U_{i}\|^{2}\equiv U^{2}
\end{align}
where $\epsilon$ has units of $cm^{2}s^{-3}$.
\end{lem}
\begin{proof}
The basic energy balance equation for a constant viscosity $\nu$ is
\begin{align}
\mathlarger{\epsilon}(t)=\frac{1}{2}\frac{d}{dt}\int_{\mathbf{R}^{3}}\|U_{i}(x,t)\|^{2}d^{3}x=-\nu\int_{\mathbf{R}^{3}}
\|\nabla_{i}U^{i}(x,t)\|^{2}d^{3}x
\end{align}
so that from (-) this is equivalent to
\begin{align}
-\int_{\mathbf{D}}\mathlarger{\epsilon}(x,t)d^{3}x=\frac{d}{dt}\int_{\mathbf{D}}U_{i}(x,t)U^{i}(x,t)d^{3}x
\end{align}
Now integrate over $[0,t]$ on both sides so that
\begin{align}
-\int_{0}^{s}\int_{\mathbf{D}}\mathlarger{\epsilon}(x,s)ds d^{3}x=\int_{\mathbf{D}}U_{i}(x,t)U^{i}(x,t)d^{3}x
\end{align}
and
\begin{align}
\left|-\int_{0}^{s}\int_{\mathbf{D}}\mathlarger{\epsilon}(x,s)ds d^{3}x\right|\equiv \left|\int_{0}^{s}\int_{\mathbf{D}}\mathlarger{\epsilon}(x,s)ds d^{3}x\right|=\left|\int_{\mathbf{D}}U_{i}(x,t)U^{i}(x,t)d^{3}x\right|
\end{align}
Now if $\epsilon(x,t)=\epsilon=const.$ and $U_{i}(x,t)=U_{i}=const. $ then
\begin{align}
\left|\mathlarger{\epsilon}\int_{0}^{s}\int_{\mathbf{D}}ds d^{3}x\right|=\left|U_{i}U^{i}\int_{\mathbf{D}}d^{3}x\right|
\end{align}
which is
\begin{align}
\left|\mathlarger{\epsilon}t Vol(\mathbf{D})\right|=\left|U_{i}U^{i}Vol(\mathbf{D})\right|
\end{align}
and so
\begin{align}
|\mathlarger{\epsilon}t|=\mathlarger{\epsilon}t=|U_{i}U^{i}|
\end{align}
\end{proof}
Checking the units on both sides gives $cm^{2}s^{-3}\times s=cm^{2}s^{-2}$, as required.
\begin{lem}
Given $|\mathlarger{\epsilon}t|=\mathlarger{\epsilon}t=|U_{i}U^{i}|$ it follows also that
\begin{align}
&\|U_{i}\|=U=\mathlarger{\epsilon}^{1/3}\ell^{1/3}\\&
\|U_{i}\|^{2}=U^{2}=\mathlarger{\epsilon}^{2/3}\ell^{2/3}\\&
\|U_{i}\|^{3}=U^{3}=\mathlarger{\epsilon}\ell
\end{align}
\end{lem}
\begin{proof}
Choose $\ell$ such that
\begin{align}
\|U_{i}\|\equiv U=\frac{\ell}{t}
\end{align}
then
\begin{align}
t=\frac{\ell}{U}
\end{align}
Then from (3.60)
\begin{align}
\mathlarger{\epsilon}t=U^{2}=\mathlarger{\epsilon}\frac{\ell}{U}
\end{align}
which gives
\begin{align}
U^{3}=\mathlarger{\epsilon}\ell
\end{align}
so that
\begin{align}
\boxed{U=\mathlarger{\epsilon}^{1/3}\ell^{1/3}}
\end{align}
Equations (3.79) and (3.80) then follow.
\end{proof}
\begin{thm}\textbf{(Emergence of a 4/5-law via an 'engineered' turbulent flow in $\mathbf{D}\subset\mathbf{R}^{3}$)}\newline
Let a vector field $ U_{i}(x,t)$ be a deterministic/smooth flow within a domain $\mathbf{D}$ of volume $Vol(\mathbf{D})\sim L^{3}$ via the Burger's equations
\begin{align}
&\partial_{t}U_{i}(x,t)+\mathlarger{\mathscr{D}}_{N}U_{i}(x,t)\nonumber\\&
=\partial_{t}U_{i}(x,t)+U^{j}(x,t)\nabla_{j}U_{i}(x,t),~~(x,t)\in\mathbf{D}\times\mathbf{R}^{+}
\end{align}
and from some initial Cauchy data $U_{i}(x,0)=g_{i}(x)$. A trivial steady state solution is then $X_{i}(x,t)=U_{i}=g_{i}=const.$. Let $\psi(x)$ be a random Gaussian scalar field as previously defined, and again having an antisymmetric rational quadratic covariance kernel
\begin{align}
\mathlarger{\mathscr{K}}(x,y;\lambda)={\mathbb{E}}[\psi(x)\psi(y)]=f(x,y)K(\|x-y\|;\lambda)
\end{align}
Here, $f(x,y)$ is an antisymmetric function $f:\mathbf{D}\times\mathbf{D}\rightarrow \lbrace 0, 1\rbrace $ such that $f(x,y)=-f(y,x)$ with $f(x,y)=1$ for all $(x,y)\in\mathbf{D}$, and $f(y,x)=-1$ with $f(x,x)=f(y,y)=0$. Then $\partial_{x}f(x,y)=\partial_{y}f(x,y)=0$. The kernel $K(\|x-y\|;\lambda)$ is again any standard stationary  and isotropic covariance kernel for Gaussian random fields; for example a rational quadratic covariance with scale-mixing parameter $\alpha$ gives
\begin{align}
\mathlarger{\mathscr{K}}(x,y;\lambda)={\mathbb{E}}[\psi(x)\psi(y)]=\beta f(x,y)\left(1-\frac{\ell^{2}}{2\alpha\lambda^{2}}\right)^{-\alpha},~~(\alpha,\beta>0)
\end{align}
Then for $y=x+\ell$
\begin{align}
&{\mathbb{E}}[\mathlarger{\psi}(x)]=0\\&
{\mathbb{E}}[\psi(x)\psi(x)]=0\\&
{\mathbb{E}}[\psi(x+\ell)\psi(x+\ell)]=0\\&
{\mathbb{E}}[\psi(x)\psi(x+\ell)]=\beta f(x,x+\ell)K(\|\ell\|;\lambda)\\&
{\mathbb{E}}[\psi(x)\psi(x)\psi(x)]=0\\&
{\mathbb{E}}[\psi(x+\ell)\psi(x+\ell)\psi(x)]=0\\&
{\mathbb{E}}[\psi(x)\psi(x)\psi(x+\ell)]=0\\&
{\mathbb{E}}[\psi(x+\ell)\psi(x+\ell)\psi(x+\ell)]=0
\end{align}
Now let $\mathbf{Q}=[0,L]$ so that
\begin{align}
\mathbf{Q}=\mathbf{Q}_{1}\bigcup\mathbf{Q}_{2}=[0,\lambda]\bigcup(\lambda,L]\nonumber
\end{align}
then either $\ell\in\mathbf{Q}_{1}$ or $\ell\in\mathbf{Q}_{2}$. We now 'engineer' the following random field representing a turbulent fluid flow within $\mathbf{D}\subset\mathbf{R}^{d}$.
\begin{align}
{\mathcal{U}}_{i}(x,t)=U_{i}(x,t)+\theta U_{i}(x,t)\mathlarger{\psi}(x)\equiv  U_{i}(x,t)+\frac{\theta}{\sqrt{d(d+2}}U_{i}(x,t)\mathlarger{\psi}(x)
\end{align}
so that $\mathbb{E}[{\mathcal{U}}_{i}(x,t)]=U_{i}(x,t)$. The 3rd-order structure function is then
\begin{align}
\mathlarger{S}_{3}(\ell)={\mathbb{E}}[\left|{\mathcal{U}}_{i}(x+\ell,t)-{\mathcal{U}}_{i}(x,t)\right|^{3}]
\end{align}
Computing $S_{3}(\ell)$ and then letting $U_{i}(x,t)\rightarrow U_{i}=(0,0,U)$, one obtains
\begin{align}
\mathlarger{S}_{3}(\ell)= -\frac{12}{d(d+2)}\theta^{2}\beta\|U_{i}\|^{3}K(\|\ell\|;\lambda)
\end{align}
In three dimensions, $d=3$ and choosing $\theta=1$ gives
\begin{align}
\mathlarger{S}_{3}(\ell)= -\frac{12}{15}\|U_{i}\|^{3}K(\|\ell\|;\lambda)=
-\frac{4}{5}\|U_{i}\|^{3}K(\|\ell\|;\lambda)
\end{align}
Choosing the kernel (2.13) with $\beta=1$
\begin{align}
\mathlarger{S}_{3}(\ell)= -\frac{12}{15}\|U_{i}\|^{3}K(\|\ell\|;\lambda)=
-\frac{4}{5}\|U_{i}\|^{3}\left(1-\frac{\ell^{2}}{2\alpha\lambda^{2}}\right)^{-\alpha}
\end{align}
Then for $\ell\in\mathbf{Q}_{1}=[0,\lambda]$, with $\ell\ll\lambda$, the term $\frac{1}{2}|\ell/\lambda|^{2}$ is very close to zero so that
\begin{align}
\mathlarger{S}_{3}(\ell)= -\frac{4}{5}\|U_{i}\|^{3}
\end{align}
holds over this range of length scales. Finally, applying (3.80)
\begin{align}
\boxed{\mathlarger{S}_{3}(\ell)= -\frac{4}{5}\|U_{i}\|^{3}=-\frac{4}{5}\mathlarger{\epsilon}\ell}
\end{align}
This is then the exact 4/5-law of turbulence.
\end{thm}
\begin{thm}
Let the scenario of Thm (2.10) and (2.12) hold. With $\mathcal{U}_{i}(x,t)$ replacing $\mathcal{X}_{i}(x,t)$, the 2nd-order structure function is given by (2.80)
\begin{align}
\mathlarger{S}_{2}(\ell)= C\|U_{i}\|^{2}
\end{align}
Then using (3.79)
\begin{align}
\boxed{\mathlarger{S}_{2}(\ell)= C\mathlarger{\epsilon}^{2/3}\ell^{2/3}}
\end{align}
which is the $2/3$ law, with $C$ a constant.
\end{thm}
\section{CONCLUSION}
The classical 2/3 and 4/5 laws of turbulence can be reproduced from a theory of engineered random fields, existing within an Euclidean domain $\mathbf{D}$. It is
admitted that the fields and their kernels, have been 'reverse-engineered' in this way so to speak, in order to get the required answers. However, this still demonstrates that the concept of random fields can lead to these important classical results when one computes their structure functions. An insight of Kolmogorov's work was that turbulent flows seem to be essentially random fields. It has been assumed that the noise or random fluctuation in fully developed turbulence is a generic noise determined by well-established general theorems in probability theory, stochastic analysis, and random fields or functions. Classical random fields or functions correspond naturally to structures, and properties of systems, that are varying randomly in time and/or space, and this should include turbulent fluids. Expectations or stochastic averages are also well defined. Rigorously defining time or spatial statistical averages in conventional statistical hydrodynamics however, is fraught with technical difficulties and limitations, as well as having a limited scope of physical applicability.
}
\end{document}